\documentclass[11pt]{article}
\usepackage[hmargin=1in,vmargin=1in]{geometry}
\usepackage{hchang}
\usepackage{thm-restate}
\usepackage{cleveref}
\usepackage{comment}
\usepackage{tcolorbox}
\usepackage{multirow}

\newtheorem{theorem}{Theorem}
\newtheorem{lemma}[theorem]{Lemma}
\newtheorem{claim}[theorem]{Claim}
\newtheorem{observation}[theorem]{Observation}

\newtheorem{definition}[theorem]{Definition}

\newcommand{\namedref}[2]{\hyperref[#2]{#1~\ref*{#2}}}
\newcommand{\Questionref}[1]{\namedref{Question }{#1}}

\usepackage{nicefrac}

\newcommand{\bag}[1]{\ensuremath{\operatorname{Bag}(#1)}}
\newcommand{\net}[1]{\ensuremath{\operatorname{Net}(#1)}}
\newcommand{\dist}{\ensuremath{d}}
\newcommand{\Texp}{\mathsf{Texp}}
\newcommand{\supp}{\mathrm{supp}}
\def\cA{\ensuremath{\mathcal{A}}}  
  
\def\cC{\ensuremath{\mathcal{C}}}

\def\cF{\ensuremath{\mathcal{F}}}

\def\cK{\ensuremath{\mathcal{K}}}

\def\cP{\ensuremath{\mathcal{P}}}
\def\cQ{\ensuremath{\mathcal{Q}}}

\def\cT{\ensuremath{\mathcal{T}}}

\def\cX{\ensuremath{\mathcal{X}}}

\newcommand{\ddim}{{\rm ddim}}
\newcommand{\tw}{{\rm tw}}
\newcommand{\pw}{{\rm pw}}
\newcommand{\Root}{\operatorname{root}}

\newcommand{\SPCS}{{\rm SPCS}}

\newcommand{\dom}{\ensuremath{\mathrm{dom}}}
\def\eps{\varepsilon}

\usepackage{subfig}
\nolinenumbers
\title{How to Protect Yourself from Threatening Skeletons:\\ Optimal Padded Decompositions for Minor-Free Graphs}

\author{%
Jonathan Conroy%
\thanks{Department of Computer Science, Dartmouth College. Email: {\tt jonathan.conroy.gr@dartmouth.edu}. Supported by the U.S.\ National Science Foundation CAREER Award under the Grant No.\ CCF-2443017.
}
\and
Arnold Filtser%
\thanks{Bar-Ilan University. Email: \texttt{arnold.filtser@biu.ac.il}. This research was supported by the ISRAEL SCIENCE FOUNDATION (grant No.\ 1042/22).
}}

\date{}
\date{}

\begin{document}
\maketitle

\begin{abstract}
Roughly, a metric space has padding parameter $\beta$ if for every $\Delta>0$, there is a stochastic decomposition of the metric points into clusters of diameter at most $\Delta$ such that every ball of radius $\gamma\Delta$ is contained in a single cluster with probability at least $e^{-\gamma\beta}$.
The padding parameter is an important characteristic of a metric space with vast algorithmic implications.
In this paper we prove that the shortest path metric of every $K_r$-minor-free graph has padding parameter $O(\log r)$, which is also tight.
This resolves a long standing open question, and exponentially improves the previous bound.
En route to our main result, we construct sparse covers for $K_r$-minor-free graphs with improved parameters, and we prove a general reduction from sparse covers to padded decompositions.

\end{abstract}

\vfill
\setcounter{secnumdepth}{5}
\setcounter{tocdepth}{3} \tableofcontents
\newpage
\pagenumbering{arabic}

\section{Introduction}
Given a metric space, a \EMPH{padded decomposition} is a stochastic partition of the space into clusters of bounded diameter such that small neighborhoods are likely to be clustered together.
Roughly, we say that a metric space has \EMPH{padding parameter} $\beta$ if it can be stochastically partitioned into clusters of (arbitrary) diameter $\Delta$, such that every ball of radius $\frac{\Delta}{\beta}$ is contained in a single cluster with probability at least $\frac12$.
From an algorithmic perspective, the padding parameter is a vitally important characteristic of a metric space.
Indeed, padded decompositions provide a natural approach for divide-and-conquer algorithms, and in numerous problems the padding parameter determines performance in the following sense:
Given a certain problem on an $n$-point metric space with padding parameter $\beta$, it is possible to achieve a solution of value $f(\beta)$ (or sometimes $g(n,\beta)$) for some function $f$ (or $g$). 
A partial list of examples include:
multi-commodity flow \cite{KPR93,LR99},
 flow sparsifiers \cite{EGKRTT14},
metric embeddings \cite{Bar96,Rao99,Rab03,KLMN04,LN05,ABN11,BFT24,Fil24sparse},
spanners \cite{HIS13,FN22,HMO23}, 
near linear SDD solvers \cite{BGKMPT14}
and other spectral methods \cite{KLPT09,BLR10}, 
Lipschitz and $0$-extension \cite{CKR04,AFHKTT04,LN05},
and many more.

We continue with a formal definition. We will use the language of shortest path metrics of finite weighted graphs $G=(V,E,w)$, as this is the topic of our paper. Note however that every finite metric space can be represented in this way. 
The \EMPH{weak diameter} of a cluster $C\subseteq V$ is the maximum pairwise distance between cluster points with respect to original distances: $\max_{u,v\in C}d_G(u,v)$.
On the other hand, the \EMPH{strong diameter} of a cluster $C$ is the maximum pairwise distance with respect to the induced cluster distances: $\max_{u,v\in C}d_{G[C]}(u,v)$.
Unless explicitly stated otherwise, in this paper we will always refer to weak diameter. 
A partition $\cP$ of $V$ is said to be \EMPH{$\Delta$-bounded} if the diameter of every cluster $C\in\cP$ is at most $\Delta$.
Given a partition $\cP$, we write \EMPH{$P(z)$} to denote the cluster containing the vertex $z$.
\begin{definition}[Padded Decomposition]\label{def:PadDecompostion}
    A distribution $\mathcal{D}$ over partitions of a graph $G=\left(V,E,w\right)$ is a \EMPH{$(\beta,\delta,\Delta)$-padded decomposition} if every $\mathcal{P}\in\supp(\mathcal{D})$ is $\Delta$-bounded, and for every $0\le\gamma\le\delta$ and $z\in V$, the ball $B_G(z,\gamma\Delta)$ satisfies
    $\Pr[B_G(z,\gamma\Delta) \subseteq P(z)] \ge e^{-\beta\gamma}$.
    We say that $G$ admits $(\beta,\delta)$-padded decomposition scheme if for every $\Delta>0$, there is a $(\beta,\delta,\Delta)$-padded decomposition for $G$.
\end{definition}

The parameter $\beta$ in \Cref{def:PadDecompostion} is called the \EMPH{padding parameter}. 
Note that as $e^{-x}\ge1-x$ (for $x\ge0$), the probability that the ball $B_G(z,\gamma\Delta)$ is cut is at most $\Pr[B_G(z,\gamma\Delta)\not\subseteq P(z)]<\beta\gamma$. 
In particular, the probability that a ball of radius $\gamma\Delta$ is contained in a single cluster goes to $1$ as $\gamma$ goes to $0$. 
The $\delta$ parameter governs range of the provided guarantee, and is always at most $\frac12$ (as the diameter of a ball with radius $\gamma\Delta$ is $2\gamma\Delta$). 
If $\delta\ge\frac1\beta$, then \Cref{def:PadDecompostion} guarantees that a ball of radius $\frac\Delta\beta$ is contained in a single cluster with probability $\frac1e$. 
Desirably, $\delta$ should be an absolute constant $\Omega(1)$.
In the discussion that follows we will ignore the $\delta$ parameter and say that a graph is $\beta$-decomposable if it admits a $(\beta,\Omega(\frac1\beta))$-padded decomposition scheme. The exact parameters are stated in \Cref{tab:padded}.

The study of padded decompositions was initiated by Klein, Plotkin, and Rao \cite{KPR93}, who showed that every $K_r$-minor free graph is $O(r^3)$-decomposable.
Later, Bartal \cite{Bar96} showed that every $n$-point metric space is $O(\log n)$-decomposable, which is also tight  \cite{Bar96} (the lower bound instance being a constant degree expander).
Note that as every $n$-vertex graph is $K_{n+1}$-minor free, this also implies an $\Omega(\log r)$ lower bound on the decomposability of $K_r$-minor free graphs.
Later, Fakcharoenphol and Talwar \cite{FT03} improved the analysis of \cite{KPR93} to show that every $K_r$-minor free graph is $O(r^2)$-decomposable.
In 2014, Abraham, Gavoille, Gupta, Neiman, and Talwar \cite{AGGNT19} showed that every $K_r$-minor free graph is $O(r)$-decomposable by introducing a decomposition based on the cops and robbers approach (completely different from \cite{KPR93}). 
Closing the exponential gap between the upper bound of $O(r)$ and the lower bound $\Omega(\log r)$ has been an outstanding problem asked by various authors \cite{FT03,Lee12,AGGNT19,Fil19padded,FIK+23}.
This is due to its fundamental nature to the understanding of the shortest path metric in minor-free-graphs, as well as its numerous algorithmic applications.

Partial progress was made in some limited special cases:
Lee and  Sidiropoulos \cite{LS10} proved that every graph of genus $g$ is $O(\log g)$-decomposable.
Abraham \etal~\cite{AGGNT19} showed that pathwidth $\pw$ graphs are $O(\log \pw)$-decomposable.
Recently, Filtser \etal~\cite{FFIKLMZ24} (improving over \cite{AGGNT19,KK17}) showed that treewidth $\tw$ graphs are $O(\log \tw)$-decomposable.
This long line of research culminates in the current paper, where we answer the main open question by providing a tight bound for $K_r$-minor free graphs.

\begin{theorem}
\label{thm:padded}
    Every $K_r$-minor-free graph admits an $(O(\log r), \Omega(1))$-padded decomposition scheme.
\end{theorem}

\begin{table}[t]\center
    \def\arraystretch{1.15}
	\begin{tabular}{|l|l|l|l|l|}
		\hline
		\textbf{Family} & \textbf{Padding parameter} & \textbf{$\delta$} & \textbf{Diameter} & \textbf{Ref} \\ \hline
		General & $O(\log n)$ & $\Omega(1)$ & strong & \cite{Bar96} \\ \hline 
		
		\multirow{5}{*}{$K_r$ minor free} & $O(r^3)$ & $\Omega(r^{-2})$ & weak & \cite{KPR93} \\ \cline{2-5} 
		& $O(r^2)$ & $\Omega(r^{-1})$ & weak & \cite{FT03} \\ \cline{2-5} 
		& $O(r)$ & $\Omega(1)$ & weak & \cite{AGGNT19} \\ \cline{2-5} 
		& $O(r^2)$ & $\Omega(r^{-2})$ & strong & \cite{AGGNT19} \\ \cline{2-5} 
		& $O(r)$ & $\Omega(r^{-1})$ & strong & \cite{Fil19padded} \\ \hline
		\multirow{2}{*}{Genus $g$} & $O(\log g)$ & $\Omega(1)$ & weak  & \cite{LS10} \\ \cline{2-5} 
  
		& $O(\log g)$ & $\Omega(1)$ & strong  & \cite{AGGNT19} \\ \hline
  
		Pathwidth $\pw$ & $O(\log\pw)$ & $\Omega(1)$ & strong %$^{(*)}$
  & \cite{AGGNT19} \\ \hline
  
		\multirow{3}{*}{Treewidth $\tw$} & $O(\log\tw+\log\log n)$ & $\Omega(1)$ & weak & \cite{KK17} \\ \cline{2-5}
  & $O(\tw)$ & $\Omega(\tw^{-1})$ & strong & \cite{AGGNT19} \\ \cline{2-5}
  & $O(\log\tw)$ & $\Omega(1)$ & weak & \cite{FFIKLMZ24} \\ \cline{2-5}

  \hline\hline
		\rowcolor[HTML]{EFEFEF}$K_r$ minor free & $O(\log r)$ & $\Omega(1)$ & weak & \Cref{thm:padded} \\ \hline
	\end{tabular}
	\caption{Summary of previous and new results on padded decompositions.}
    \label{tab:padded}
\end{table}

\subsection{Sparse Cover}
\EMPH{Sparse cover} is in some sense a dual notion to padded decompositions. Here we have a collection of (non-disjoint) bounded-diameter clusters such that every small enough ball is fully contained in some cluster, while every vertex belongs only to a limited number of clusters.

\begin{definition}[Sparse Cover]\label{def:SparseCover}
	Given a weighted graph $G=(V,E,w)$, a collection of clusters $\mathcal{C} = \{C_1,..., C_t\}$ is a \EMPH{$(\beta,s,\Delta)$ sparse cover} if:
	\begin{enumerate}
		\item Bounded diameter: The diameter of every cluster $C_i\in\mathcal{C}$ is bounded by $\Delta$.\label{condition:RadiusBlowUp}
		\item Padding: For each $v\in V$, there exists a cluster $C_i\in\mathcal{C}$ such that $B_G(v,\frac\Delta\beta)\subseteq C_i$.
		\item Sparsity: For each $v\in V$, there are at most $s$ clusters in $\mathcal{C}$ containing $v$.		
	\end{enumerate}	
	If the clusters $\cC$ can be partitioned into $s$ partitions $\cP_1,\dots,\cP_s$ such that $\cC=\cup_{i=1}^s\cP_i$, then  $\{\cP_1,\dots,\cP_s\}$ is called a \EMPH{$(\beta,s,\Delta)$ sparse partition cover}. 
We say that a graph $G$ admits a  $(\beta,s)$ sparse (partition) cover scheme, if for every parameter $\Delta>0$ it admits a $(\beta,s,\Delta)$ sparse (partition) cover that can be constructed in expected polynomial time. Sparse partition cover scheme is abbreviated \SPCS. 
\end{definition}

If all the clusters in the sparse cover have strong diameter guarantee, then we will call it a strong sparse cover (respectively, strong sparse cover scheme and strong \SPCS).
Sparse covers were introduced by Awerbuch and Peleg \cite{AP90} (even before padded decompositions), and have found numerous applications. 
A partial list includes: 
compact routing schemes \cite{PU89,Peleg00,TZ01b,AGMNT08,AGM05,AGMW10,BLT14},
distant-dependent distributed directories \cite{AP91,Peleg93,Peleg00,BLT14},
network synchronizers \cite{AP90b,Lynch96,Peleg00,AW04,BLT14},
distributed deadlock prevention \cite{AKP94},
construction of spanners and ultrametric covers \cite{HIS13,FN22,LS23,HMO23,FL22,FGN24},
metric embeddings \cite{Rao99,KLMN04,Fil24sparse},
universal TSP and Steiner tree constructions \cite{JLNRS05,BDRRS12,Fil24scattering,BCFHHR24}, and oblivious buy-at-bulk \cite{SBI11}.

Given a padded decomposition one can construct a sparse cover. Indeed, suppose that $G$ admits a  $(\beta,\delta,\Delta)$-padded decomposition. Fix $\rho\ge\frac1\delta$, and take all the clusters in a union of $e^{\frac\beta\rho}\cdot O( \log n)$ independent samples from the padded decomposition. With high probability, the resulting collection of clusters will be a $(\rho,e^{\frac\beta\rho}\cdot O( \log n),\Delta)$-sparse cover.
However, no similar statement in the other direction was ever made.
The second contribution of the current paper is to show that in general, given a sparse cover one can construct a padded decomposition.\footnote{\Cref{thm:coversToPaddedDecomp} was previously included in a manuscript \cite{Fil24sparse} uploaded to arXiv. The reduction has since been removed from that manuscript and now appears exclusively here.}

\begin{restatable}[From sparse covers to padded decompositions]{theorem}{SparseCoverToDecomposition}
	\label{thm:coversToPaddedDecomp}
	Consider a weighted graph $G=(V,E,w)$ that admits a $(\beta,s,\Delta)$-sparse cover. Then $G$ admits $(O(\beta\cdot\log s),\frac{1}{4\beta},\Delta)$-padded decomposition.
\end{restatable}

\begin{table}[]\center
    \def\arraystretch{1.15}
	\begin{tabular}{|l|l|l|l|l|}
		\hline
		\multicolumn{1}{|l|}{\textbf{Family}} & \textbf{Padding} & \textbf{Sparsity}                          & \textbf{Diameter} & \textbf{Ref}       \\ \hline
		\multicolumn{1}{|l|}{General}         & $4k-2$           & $2k\cdot n^{\nicefrac1k}$                      & strong            & \cite{AP90}               \\ \hline
		Planar&$32$&$18$ &strong&\cite{BLT14}\\\hline
		\multirow{8}{*}{$K_r$-minor free}     & $O(r^3)$         & $O(\log n)$                                & weak              & \cite{KPR93}              \\ \cline{2-5} 
		& $O(r^2)$         & $O(\log n)$                                & weak              & \cite{FT03}               \\ \cline{2-5} 
		
		% & $O(r)$         & $O(\log n)$                                & weak              & \cite{AGGNT19}               \\ \cline{2-5} 
		
		& $O(r)$         & $O(\log n)$                                & strong              & \cite{Fil19padded}               \\ \cline{2-5} 
		
		& $8$              & $O_r(\log n)$                       & strong           & \cite{BLT14}              \\ \cline{2-5} 
		& $O(r^2)$         & $2^r$                                      & weak              & \cite{FT03,KLMN04}             \\ \cline{2-5} 
		& $O(r^2)$         & $2^{O(r)}\cdot r!$                                           & strong            & \cite{AGMW10}             \\ \cline{2-5} \cline{2-5} 
		
		& $O(r)$             & $O(r^2)$                                  & strong            & \cite{Fil24sparse}  \\ \cline{2-5} 		
		& $4+\eps$             & $O(\frac1\eps)^r$                                   & strong            & \cite{Fil24sparse}\\ \hline
		Treewidth $\tw$       & $6$              & $\poly(\tw)$ &weak&\cite{FFIKLMZ24}\\\hline  \hline   
		\rowcolor[HTML]{EFEFEF}$K_r$-minor free       & $8+\eps$              & $O(r^4/\eps^2)$ &weak&\Cref{thm:cover}\\\hline 
		
	\end{tabular}
	\caption{\small{Summery of new and previous work on sparse covers. 
 The $O_r$ notation for \cite{BLT14} hides the constant from the Robertson Seymour \cite{RS03} structure theorem$^{\ref{foot:RSconstant}}$.}}
	\label{tab:Covers}
\end{table}	

See \Cref{tab:Covers} for a summary of previous work on sparse covers. 
Awerbuch and Peleg \cite{AP90} showed that for $k\in\N$, general $n$-vertex graphs admit a $(4k-2,2k\cdot n^{\frac1k})$ sparse cover scheme.
In the context of $K_r$-minor-free graphs, using \cite{KPR93} and the reduction above one can get a cover with padding $\poly(r)$ and $O(\log n)$ sparsity. 
Krauthgamer \etal~\cite{KLMN04} showed that, using \cite{KPR93}, one can get sparse covers with parameters independent from the cardinally of the graph; specifically, \cite{KPR93} can be transformed into an $(O(r^2),2^r)$-sparse cover scheme.
Recently, Filtser \cite{Fil24sparse} showed that the cops and robbers approach of \cite{AGGNT19} (more specifically, its refinement recently developed in \cite{CCLMST24}) can be transformed into an $(O(r),O(r^2))$-sparse cover scheme, and $(4+\eps,O(\frac1\eps)^r)$-sparse cover scheme, both with strong diameter.
Even later,  Filtser \etal~\cite{FFIKLMZ24} showed that every graph with treewidth $\tw$ admits a $(6,\poly(\tw))$-sparse cover scheme.
The third contribution of this paper is a much improved sparse covers for minor-free graphs. 
\begin{theorem}\label{thm:cover}
    Every $K_r$-minor-free graph $G$ admits an $\left(8+\eps, O\left(\frac{r^4}{\eps^2}\right)\right)$-\SPCS~%sparse cover scheme 
    for every $\eps>0$.
\end{theorem}
Note that for constant padding, \Cref{thm:cover} is an exponential improvement compared to the previous state of the art \cite{Fil24sparse}.
As a \SPCS~(sparse partition cover scheme) is in particular a sparse cover, by combining this improved sparse cover (\Cref{thm:cover}) with the reduction from \Cref{thm:coversToPaddedDecomp}, our main \Cref{thm:padded} follows.
Our \Cref{thm:padded} and \Cref{thm:cover} have numerous algorithmic applications. We mention some of them in \Cref{sec:apps}.

\subsection{Related Work and Additional Background}
\paragraph*{Minor structure theorem.} In their celebrated work on the structure theorem, Robertson and Seymour \cite{RS03} showed that minor-free graphs can be decomposed into ``basic components”: surface-embedded graphs, apices, vortices and clique-sums.
This decomposition provides an algorithmic methodology: solve the problem on planar graphs, and then generalize to richer structure step by step until  finally we get to minor-free graphs (see e.g. \cite{CFKL20} for an example).
Alas, the constants hiding in the structure theorem \cite{RS03} are enormous%
\footnote{Johnson \cite{Johnson87} estimated that the constant hiding in the structure theorem of \cite{RS03} is larger than $2\Uparrow(2\Uparrow((2\Uparrow\frac r2)+3))$ where $2\Uparrow t$ is the exponential tower function ($2\Uparrow0=1$ and $2\Uparrow t=2^{2\Uparrow(t-1)}$).\label{foot:RSconstant}}, 
making any algorithm following this path completely impractical. 
A significant advantage of the padded decomposition/sparse cover framework for algorithm design is that there are no enormous hidden constants, and the dependence on the minor size $r$ is reasonable. In this work we get the best possible dependence on $r$.

\paragraph*{Strong diameter} Recall that the strong diameter of a cluster $C$ is $\max_{u,v\in C}d_{G[C]}(u,v)$ the maximum pairwise distance in the induced graph.
While this paper is concerned with weak diameter, it is often more convenient to work with strong diameter, and in fact some applications indeed require strong diameter guarantee (e.g. for routing, spanners, etc.).
Filtser \cite{Fil19padded} (improving over \cite{AGGNT19}) showed that $K_r$-minor free graphs admit strong $(O(r),\Omega(r^{-1}))$-padded decomposition scheme. This remains the state of the art.
There been many previous works on strong sparse covers (see \Cref{tab:Covers} for summary).

\paragraph*{Hierarchies.} Padded decompositions were used to construct stochastic tree embeddings \cite{Bar96,FRT04,AN19}. Specifically, in \cite{Bar96} one samples independently padded decompositions in all possible scales and combines them into an HST to get expected distortion $O(\log^2n)$.
In follow-up works \cite{Bartal98,FRT04,Bartal04}, padded decomposition in different scales are sampled in a correlated manner to obtain the optimal $O(\log n)$ expected distortion (see \cite{AN19} for a strong diameter counterpart).
A stronger version of hierarchical padded decomposition was studied in the context of Ramsey type embedding, where with some probability a single vertex is padded in all possible scales simultaneously (see e.g. \cite{MN07,BGS16}, and \cite{ACEFN20} for a strong diameter counterpart).
Hierarchical sparse covers were studied in the context of clan embeddings (or multi-embeddings) \cite{BM04,FL21,Bar21,Fil21}, ultrametric covers \cite{FL22,Fil23,FGN24}, and sparse partitions \cite{BDRRS12,BCFHHR24}.

\paragraph*{Other metric spaces.}
Metric spaces with doubling dimension $\ddim$ are $O(\ddim)$-decomposable \cite{GKL03,ABN11,Fil19padded} and admit $(O(t),2^{\frac{\ddim}{t}}\cdot\ddim\cdot\log t)$ sparse covers (for arbitrary $t\ge 1$) \cite{Fil19padded}.
Interestingly, while this is tight even for $d$-dimensional Euclidean space $\R^d$, if one is interested only in a pair of vertices (instead of a ball) belonging to the same cluster, the padding parameter is improved to $O(\sqrt{d})$ \cite{CCGGP98}.
Recently it was shown that metrics with highway dimension $h$ are $O(\log h)$-decomposable \cite{FF25}.

\section{Technical Ideas}

\paragraph*{Cop Decomposition.} 
The starting point of our story is the padded decomposition by Abraham, Gavoille, Gupta, Neiman, and Talwar \cite{AGGNT19}: inspired by the cops-and-robbers technique of \cite{And86}, they construct a \emph{cop decomposition} for minor-free graphs.%
\footnote{Our presentation here actually follows the interpretation from \cite{Fil19padded}.}
This process creates \EMPH{supernodes} with special structure (rather than immediately constructing diameter $\Delta$ clusters). 
To create the first supernode, choose an arbitrary vertex $v_1$ and sample a radius $r_1\in[0,\Theta(\Delta)]$ using truncated exponential distribution (that is, exponential distribution conditioned on the value being in $[0,\Theta(\Delta)]$).
The first supernode is simply the ball $\eta_1\coloneqq B(v_1,r_1)$.
For the second supernode, choose an unclustered vertex $v_2$, and let $T_{\eta_2}$ be a shortest path from $v_2$ to a vertex neighboring $\eta_1$. 
The path $T_{\eta_2}$ is called the \EMPH{skeleton} of $\eta_2$. 
The supernode $\eta_2 \coloneqq B_{G\setminus\eta_1}(T_{\eta_2},r_2)$ is then created as a ball around its skeleton (with  radius $r_2$ sampled as previously).
In general, after creating supernodes $\eta_1,\dots,\eta_{i-1}$, we create the $i$'th supernode \EMPH{$\eta_i$} as follows: Choose an arbitrary connected component of unclustered vertices, denoted \EMPH{$\dom(\eta_i)$}, or the \EMPH{domain} of \EMPH{$\eta_i$}.
Let \EMPH{$\cK_{\eta_i}$}$=\{\eta_{j_1},\eta_{j_2},\dots,\eta_{j_k}\}$ be all the previously created supernodes with neighbors in $\dom(\eta_i)$.
The skeleton of $\eta_i$, denoted \EMPH{$T_{\eta_i}$}, is a shortest path tree in $\dom(\eta_i)$ rooted at an arbitrary vertex $v_i\in\dom(\eta_i)$ such that for each supernode $\eta_j\in\cK_{\eta_i}$, the skeleton $T_{\eta_i}$ contains some vertex neighboring $\eta_j$. 
In particular, the skeleton tree $T_{\eta_i}$ has at most $|\cK_{\eta_i}|$ leaves (not counting the root).
The supernode $\eta_i \coloneqq B_{\dom(\eta_i)}(T_{\eta_i},r_i)$ is than created as a ball around its skeleton (with radius $r_i$ sampled as previously). 
See \Cref{fig:CopDecompAnimation} for illustration of this process.
This process recursively partitions all the vertices of the graph into supernodes, and naturally induces a tree over the supernodes denoted \EMPH{$\cT_G$} (where $\eta_i$ is a child of the most recently created supernode in $\cK_{\eta_i}$). We call $\cT_G$ a \EMPH{partition tree}.
Additionally, with each supernode $\eta_i$, we associate a set of supernodes $\EMPH{$\bag{\eta_i}$} \coloneqq \cK_{\eta_i} \cup \set{\eta_i}$.
There is a naturally induced tree decomposition of the graph $G$ with the same structure as the partition tree $\cT_G$, where the bag associated with $\eta_i\in \cT_G$ consist of all the vertices in the supernodes of $\bag{\eta_i}$; see \Cref{fig:CopDecomp}.

\begin{figure}[h!]

\includegraphics[width=.24\textwidth,page=3]{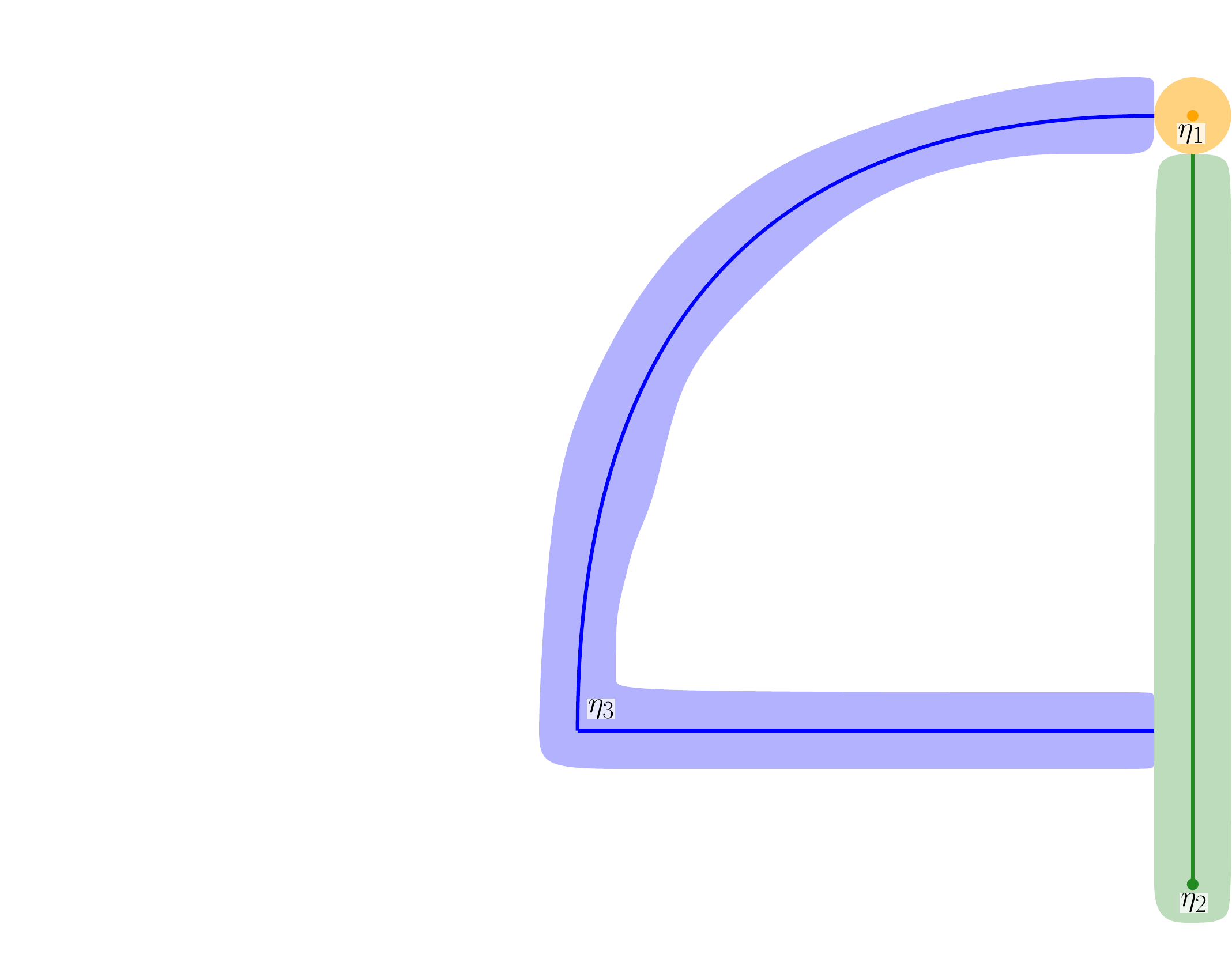}~~          \includegraphics[width=.24\textwidth,page=4]{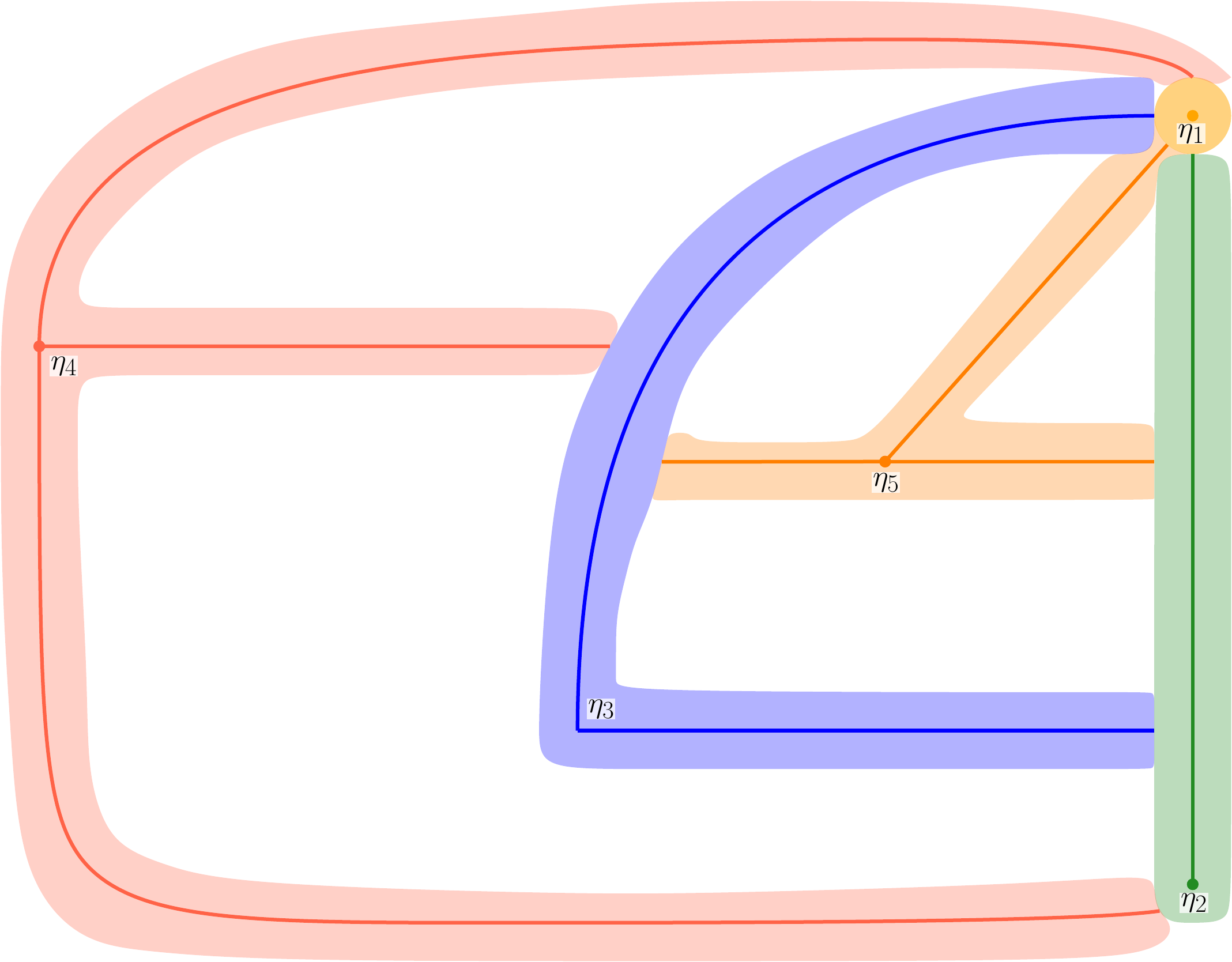}~~          \includegraphics[width=.24\textwidth,page=5]{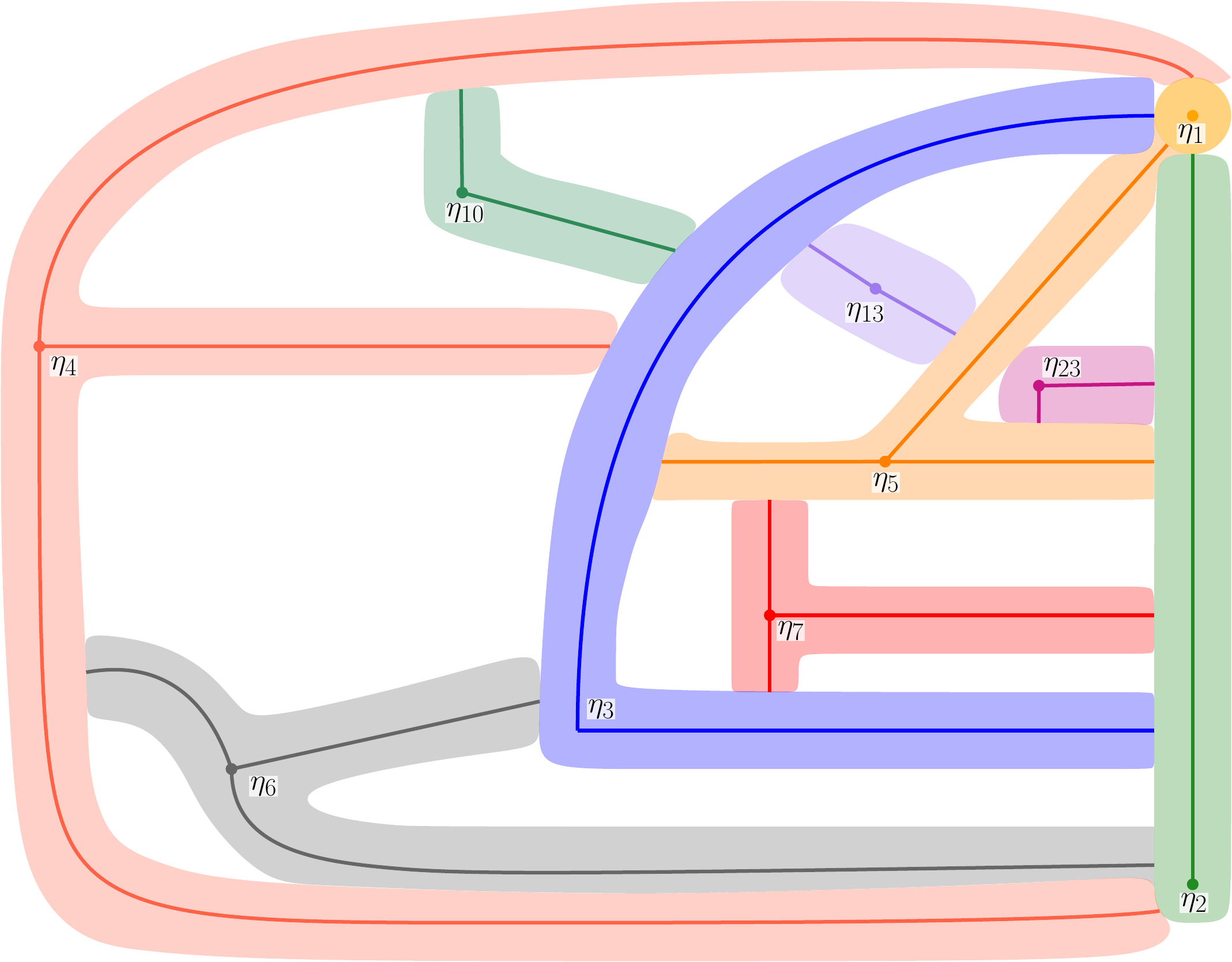}~~  \includegraphics[width=.24\textwidth,page=6]{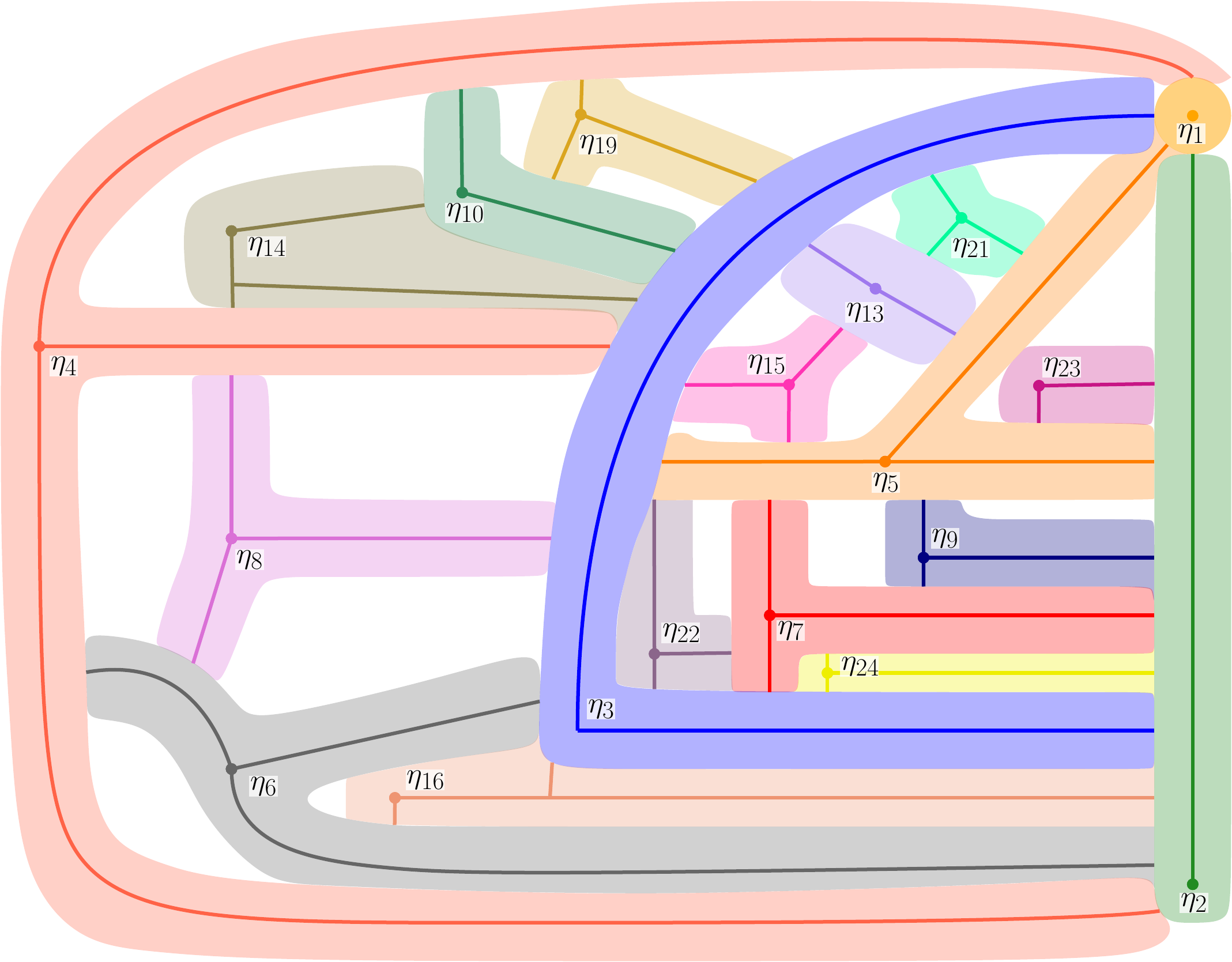}

	\caption{\textit{An animation of constructing supernodes on an (implict) planar graph, following the cop decomposition of \cite{AGGNT19}.
    Each supernode $\eta_i$ is constructed in a connected component $\dom(\eta_i)$ of $G\setminus\cup_{j<i}\eta_j$.
    The skeleton $T_{\eta_i}$ of $\eta_i$ is a shortest path tree in $\dom(\eta_i)$, with at least one vertex neighboring each previously-created supernode adjacent to $\dom(\eta_i)$. 
    In \cite{AGGNT19} the supernode $\eta_i$ is a ball around $T_{\eta_i}$ of radius at most $\Delta$. Later, we discuss buffered cop decomposition \cite{CCLMST24} which is a similar object (but with a different construction); there, $\eta_i$ is
    a connected set of vertices at distance $\le \Delta$ around $T_{\eta_i}$.}
    }
	\label{fig:CopDecompAnimation}
\end{figure}

\begin{figure}[h!]
	\begin{center}
         \includegraphics[width=.6\textwidth]{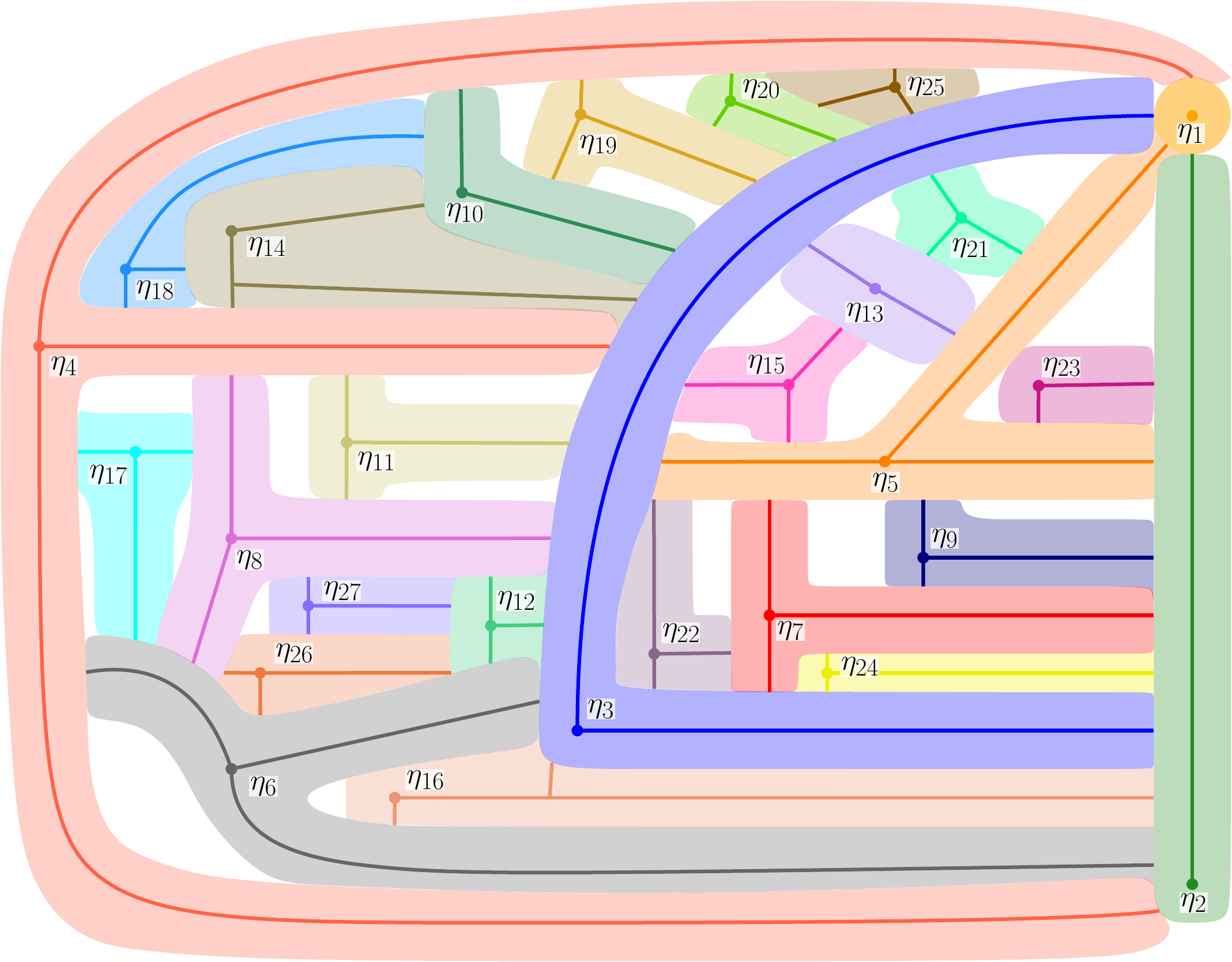}

        \vspace{1em}
         
        \includegraphics[width=\textwidth]{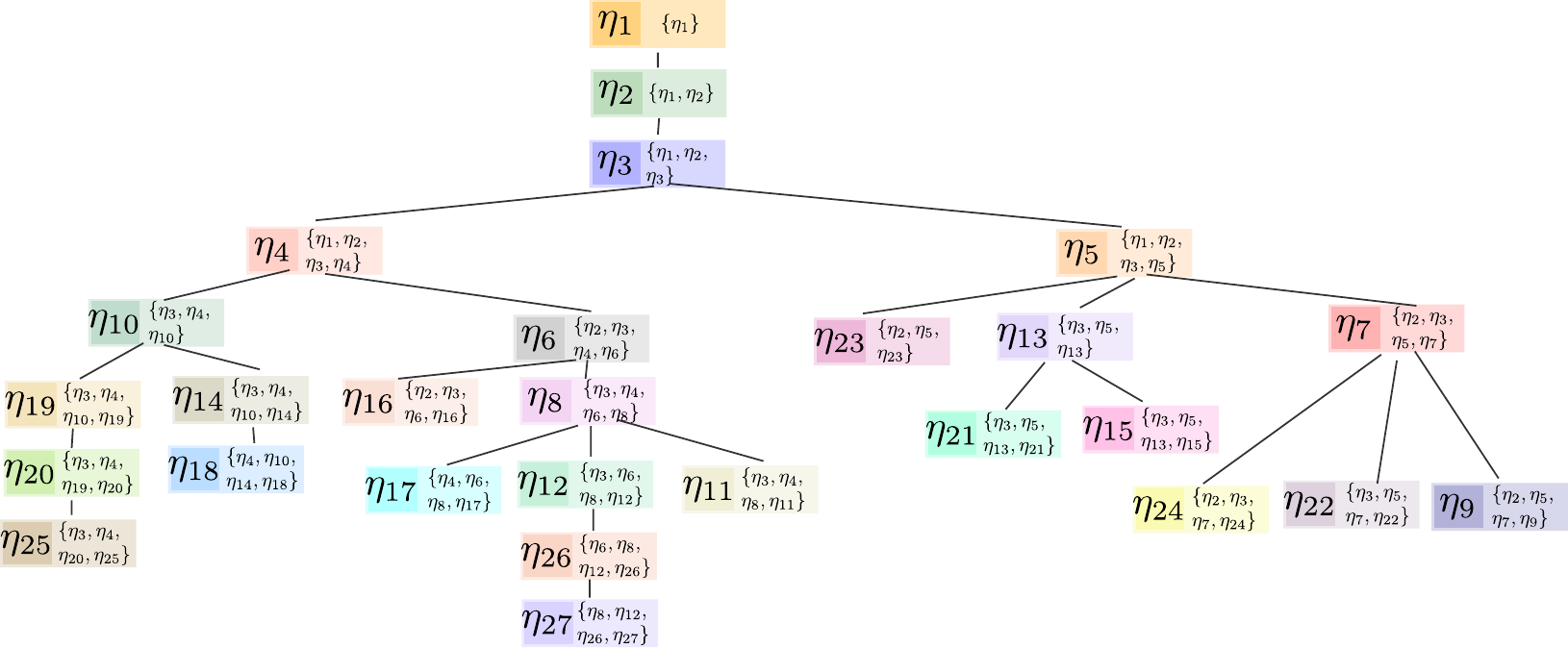} 
	\end{center}
	\caption{\textit{Top: Graph $G$ partitioned into supernodes of a cop decomposition.
    Bottom: Partition tree $\cT_G$ of the cop decomposition.
    The supernode $\eta_i$ is the child of the most-recently created supernode adjacent to $\dom(\eta_i)$.
    The set $\bag{\eta_i}$ contains $\eta_i$ and all previously-created supernodes adjacent to $\dom(\eta_i)$; it is illustrated to the right of each $\eta_i$ in the partition tree.
    The bags induce a tree decomposition of $G$.
    }}
	\label{fig:CopDecomp}
\end{figure}

The cop decomposition can be executed on an arbitrary graph. The property that makes it especially interesting for $K_r$-minor-free graphs is that the number of previously-created neighboring supernodes $k \coloneqq |\cK_{\eta_i}|$ is necessarily bounded by $r-2$ --- indeed, by contracting all the internal edges in $\eta_i$, and each of the supernodes in $\cK_{\eta_i}$ we will obtain a $K_{k+1}$ minor.
In particular, the skeleton $T_{\eta_i}$ is a shortest path tree with at most $r-2$ leaves.
While the supernodes themselves don't have any bound on their diameter, the existence of the skeleton guarantees that each supernode admits an $\left(O(\log r),\Omega(1),\Delta\right)$-padded decomposition \cite{AGGNT19,Fil19padded}.
As one can concatenate padded decompositions, in order to get a padded decomposition for the entire graph it is enough to analyze the probability that a fixed ball is fully contained in a single supernode.
Fix a ball $\EMPH{$B$} \coloneqq (v,\gamma\Delta)$.
A skeleton $T_{\eta_i}$ is called a \EMPH{threatening skeleton} if it has a positive probability to intersect $B$ (that is, if the distance from $T_{\eta_i}$ to $B$ in $\dom(\eta_i)$ is at most $O(\Delta)$).
Abraham \etal~\cite{AGGNT19} showed that the expected number of threatening skeletons is $2^{O(r)}$.
As a rule of thumb, the padding parameter is logarithmic in the number of threatening events,\footnote{Roughly speaking, if we were to use real exponential distribution to sample the supernodes radii, then the padding parameter would equal to the parameter of the exponential distribution. 
Truncated exponential distribution behaves similarly to real exponential distribution up to a small error. The errors from the samples of the different threateners accumulate, and in order to make their sum negligible, the parameter has to be at least logarithmic in the number of threateners.} 
and thus \cite{AGGNT19} were able to obtain $O(r)$ padding parameter.
Thus to improve the padding parameter one have to drastically reduce the number of threateners. Numerous attempts (in particular by the authors) were made to introduce different changes to the cop decomposition process in order to reduce the number of threateners, but to no avail.%
\footnote{A slightly different approach was presented in the full version of \cite{Fil19padded}. Here Filtser used real exponential distribution with parameter $\lambda$. It is then straightforward that the resulting padding parameter is $\Theta(\lambda)$. However, the radii of the supernodes is not bounded anymore. To fix this, at the end of the process, Filtser recurses on each supernode with radius larger than $\Omega(\Delta)$. For $\lambda=O(r)$ it then holds that a fixed vertex belongs to a supernode with radius $O(\Delta)$ with constant probability.}

\paragraph*{Padding Decomposition via Sparse Covers.}
When creating padded decomposition by iteratively carving balls (as in the cop decomposition process described above), each decision is irrevocable.
Sparse cover, on the other hand, is a much more ``forgiving'' object, as one might hope to cover a vertex with multiple different balls without making irrevocable decisions.
To exploit this, 
we reduce padded decomposition to sparse cover.

Filtser \cite{Fil19padded} proved that the existence of a ``sparse net'' implies padded decomposition. 
Specifically, consider a scenario where there is a subset $N\subseteq V$ of vertices (called a \EMPH{net}) such that every vertex $v$ has a net point at distance at most $d_G(v,N)\le \Delta$, and the number of net points in the ball $B_G(v,3\Delta)$ is at most $\tau$.
Here one can create clusters around the net points (for example curve balls with radii in $[\Delta,2\Delta]$ sampled using truncated exponential distribution).
Then for a ball $B=B(v,\gamma\Delta)$, all the ``threateners'' are net points at distance at most $(2+\gamma)\Delta$ from $v$, and in particular there are at most $\tau$ threateners.
Using a proof based on the rule of thumb from above, Filtser showed that such a net implies a $(O(\log \tau),\Omega(1),\Theta(\Delta))$-padded decomposition.
Now, suppose that we are given a graph with a $(\beta,s,\Delta)$-sparse cover $\cC$. 
Construct an auxiliary graph where for every cluster $C\in \cC$ we add an auxiliary vertex $v_C$ with an edge towards every vertex $u\in C$ of weight $w(\{v_{C},u\})=\max\left\{ \Delta,2\Delta-d_{G}(u,V\setminus C)\right\}$. 
Note that the weight of this edge decreases as the distance from $u$ to the boundary of $C$ increases.
In particular, every vertex $u$ will be at distance at most $2\Delta-\frac{\Delta}{\beta}$ from a cluster $C$ where it is padded, and at distance at least $2\Delta$ from any cluster $C'$ not containing $u$. 
Thus there are at most $\tau$ auxiliary vertices at distance strictly below $2\Delta$.
This gap then allows us to construct a padded decomposition where the auxiliary vertices are used as the net above, and thus each ball has at most $\tau$ threateners.
Morally, this is already enough to prove our \Cref{thm:coversToPaddedDecomp}.%
\footnote{Actually, in our proof we don't construct this auxiliary graph, as for technical reasons we cannot use \cite{Fil19padded} as a black box. Instead, we directly create clusters using the exponential clocks clustering of \cite{MPX13}.}

\paragraph*{Buffered Cop Decomposition.} 
Chang, Conroy, Le, Milenkovi\'{c}, Solomon, and Than
\cite{CCLMST24} introduced a deterministic construction of the cop decomposition with an additional \EMPH{buffer property}.
Their construction follows similar lines to \cite{AGGNT19}, except that the supernode $\eta_i$ is not simply a random ball around the skeleton $T_{\eta_i}$; rather, $\eta_i$ is a delicately-chosen subset of vertices in $\dom(\eta_i)$ at distance at most $\Delta$ from $T_{\eta_i}$.
We do not describe the specifics of their construction, which is somewhat involved, but instead use their result as a black box.
Chang \etal~obtain the following buffer property: for every ancestor supernode $\eta_j$ of $\eta_i$
(that is, such that $\eta_i\subset\dom(\eta_j)$), if $\eta_j\notin \bag{\eta_i}$ then $\dist_{\dom(\eta_j)}(\eta_j,\eta_i)\ge\frac{\Delta}{r}$.
In words, either there is a vertex in $\dom(\eta_i)$ (and thus also in $T_{\eta_i}$) with a neighbor in $\eta_j$, or the distance from $\eta_j$ to $\eta_i$ w.r.t. the domain of $\eta_j$ (the connected component where $\eta_j$ is created) is at least $\frac{\Delta}{r}$.
The randomized construction of \cite{AGGNT19} achieves something similar to this buffer property \emph{in expectation}%
\footnote{For the sake of intuition, fix a vertex $v$, and suppose that we were to sample the radii of supernodes using exponential distribution (as opposed to truncated exponential distribution). Consider a situation where we create the supernode $\eta_i$, where  $v\in\dom(\eta_i)$, and $\dom(\eta_i)$ is broken into several connected components, such that the connected component $\cA$ containing $v$ is disconnected from $\eta_j\in\bag{\eta_i}$. Then as the exponential distribution is memoryless, in expectation $r_i$ will grow by additional $\frac\Delta r$ factor, insuring that $\mathbb{E}[\dist_{\dom(\eta_i)}(\eta_j,v)]\ge\frac\Delta r$. The contribution of \cite{CCLMST24} is to choose $\eta_i$ (as well as other supernodes) so that for every such vertex $v$, deterministically $\dist_{\dom(\eta_j)}(\eta_j,v)]\ge\frac\Delta r$.}
and this is how they prove that the \emph{expected} number of threatening skeletons is at most $2^{O(r)}$.
See \Cref{def:buffered-cop} for a formal definition of buffered cop decomposition.

Filtser \cite{Fil24sparse} used the buffered cop decomposition of \cite{CCLMST24} to construct a $(O(1),O(1)^r)$ sparse cover scheme.
Specifically, for every supernode $\eta_i\in\cT_G$, construct a cluster $C_i \coloneqq B_{\dom(\eta_i)}(\eta_i,3\Delta)$.
While there is no bound on the diameter of the cluster $C_i$, one can use the skeleton $T_{\eta_i}$ to construct an $(\Omega(1),O(r),O(\Delta))$ sparse cover for $C_i$. Thus it is enough to prove the padding and sparsity properties with respect to the clusters $\{C_i\}_{\eta_i\in\cT_G}$.
To see the padding property, fix a ball $\EMPH{$B$}\coloneqq B(v,\Delta)$ and consider the supernode $\eta_i$ of minimum depth (with respect to $\cT_G$, i.e. the one created first) containing any vertex from $B$. 
By minimality and the triangle inequality, it holds that $B\subseteq C_i$. For sparsity, consider a vertex $v$ belonging to some supernode $\eta_i$. 
Using the deterministic $\frac\Delta r$ buffer, Filtser showed that there are at most $O(1)^r$ ancestor supernodes of $\eta_i$ at distance less than $3\Delta$ from $\eta_i$. It follows that $v$ is contained in only $O(1)^r$ clusters.
By applying our reduction (\Cref{thm:coversToPaddedDecomp}) on the sparse cover of \cite{Fil24sparse}, we obtain a padded decomposition with padding parameter $O(r)$, recovering the result of \cite{AGGNT19}. To go beyond that, we need better sparse covers.

\paragraph*{First Attempt: Sparse Covers Using Centroids.}
We describe a divide-and-conquer algorithm to construct a sparse cover with padding $O(1)$ and sparsity $\poly(r) \cdot \log n$. Consider the cop decomposition $\cT_G$. For every supernode $\eta$, $\bag{\eta}$ is a separator for $G$ consisting of at most $r-1$ supernodes. This suggests that a divide-and-conquer approach may be possible. Let \EMPH{$X$} be the \EMPH{centroid} supernode in $\cT_G$ 
(that is, every connected component of $\cT_G\setminus X$ contains at most $|\cT_G|/2$ supernodes)\footnote{While we usually denote supernodes with $\eta$, we will use $X$ to denote the centroid supernode.};
for every supernode $X' \in \bag{X}$, grow a cluster around $X'$; and then for each connected component $\cT_i$ of $\cT_G\setminus  X$, recurse on the graph induced by the vertices in the supernodes in $\cT_i$. Unfortunately, this doesn't work: because the skeleton $T_{X'}$ is only a shortest path tree in $\dom(X')$, we must only grow a cluster around $X'$ in $\dom(X')$. But growing clusters in $\dom(X')$ is not sufficient --- vertices of a single ball $B \coloneqq B(v, \Delta)$ might belong to both ancestor and descendant supernodes of $X'$; such a ball $B$ will neither be contained in the cluster around $X'$, nor in any graph induced by the supernodes in a connected component of $\cT_G \setminus X$. See \Cref{fig:recursive-cover-simple} for an illustration of this incorrect algorithm, and why it fails.
\begin{figure}[h!]
    \centering
    \includegraphics[width=0.65\linewidth]{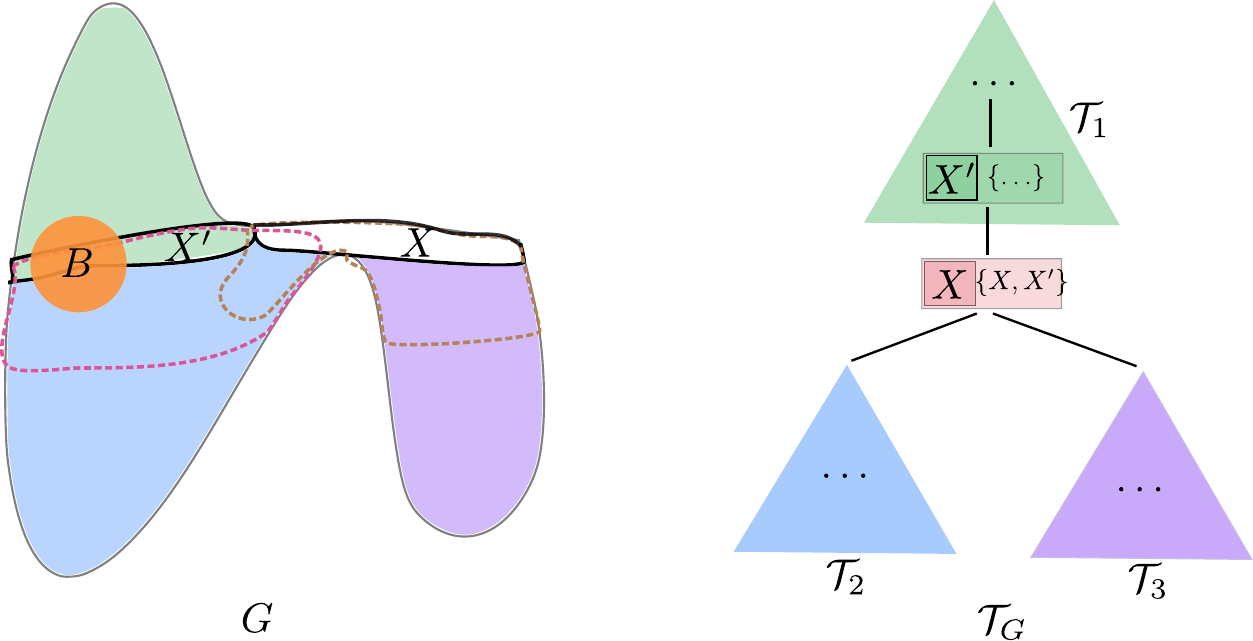}
    \caption{\textit{A styled depiction of the incorrect sparse cover algorithm on a graph $G$ (left) and its cop decomposition $\cT_G$ (right).
    On the right: Supernode $X$ is highlighted in red. Deleting $X$ from $\cT_G$ creates three connected components $\cT_1$, $\cT_2$, and $\cT_3$.
    On the left: The supernodes $\set{X, X'} = \bag{X}$ form a separator in $G$. The subgraphs induced by $\cT_1$, $\cT_2$, and $\cT_3$ are highlighted in $G$. Red and brown dashed lines outline a cluster grown around $X$ in $\dom(X)$ and a cluster grown around $X'$ in $\dom(X')$. The ball $B$ is not contained in either of the clusters, nor is it contained in a subgraph that is dealt with recursively.
    }}
    \label{fig:recursive-cover-simple}
\end{figure}

Perhaps surprisingly, one can fix this issue with a small tweak; the idea is illustrated in \Cref{fig:recursive-cover}. In the recursion, instead of simply recursing on a subtree $\cT$ of $\cT_G$, we will keep track of both a \EMPH{subtree $\cT$} of $\cT_G$ and a subset of \EMPH{active vertices $A$}. 
Crucially, the vertices in $A$ do not necessarily belong to a supernode in $\cT$, but every active vertex is within distance $2 \Delta$ of some supernode in $\cT$.
We construct a sparse cover for the active vertices $A$, using the tree $\cT$.
(Each created cluster will contain only active vertices.)
As before, we select a centroid supernode $X$ from $\cT$ and grow a cluster $C_{ X'} \coloneqq B_{\dom( X')}( X',4\Delta)\cap A$ around every supernode $X' \in \bag{X}$.
We will then recurse on every connected component $\cT_i$ of $\cT\setminus X$, after selecting a new set of active vertices $A_i$. 
The active vertex sets of the different connected components of $\cT\setminus X$ will be disjoint.
For an active vertex $v\in A$, let \EMPH{$\eta[v]$} be the supernode of minimum depth (with respect to $\cT$) such that $\dist_{\dom(\eta[v])}(\eta[v],v)\le2\Delta$.
Each vertex $v$ will join the set $A_i$ associated with the connected component $\cT_i$ of $\cT\setminus X$ that contains $\eta[v]$ (or $v$ will join no active set if $\eta[v]= X$).
We recursively create a sparse cover $\cC_i$ for each connected component $\cT_i$ and active set $A_i$.
The final cover will consist of the clusters 
% $C_{ X}$,
$\left\{C_{ X'}\right\}_{ X'\in\bag{ X}}$ together with the clusters in $\cC_i$ for each recursive call; see \Cref{fig:recursive-cover}.

\begin{figure}[h!]
    \centering
    \includegraphics[width=0.65\linewidth]{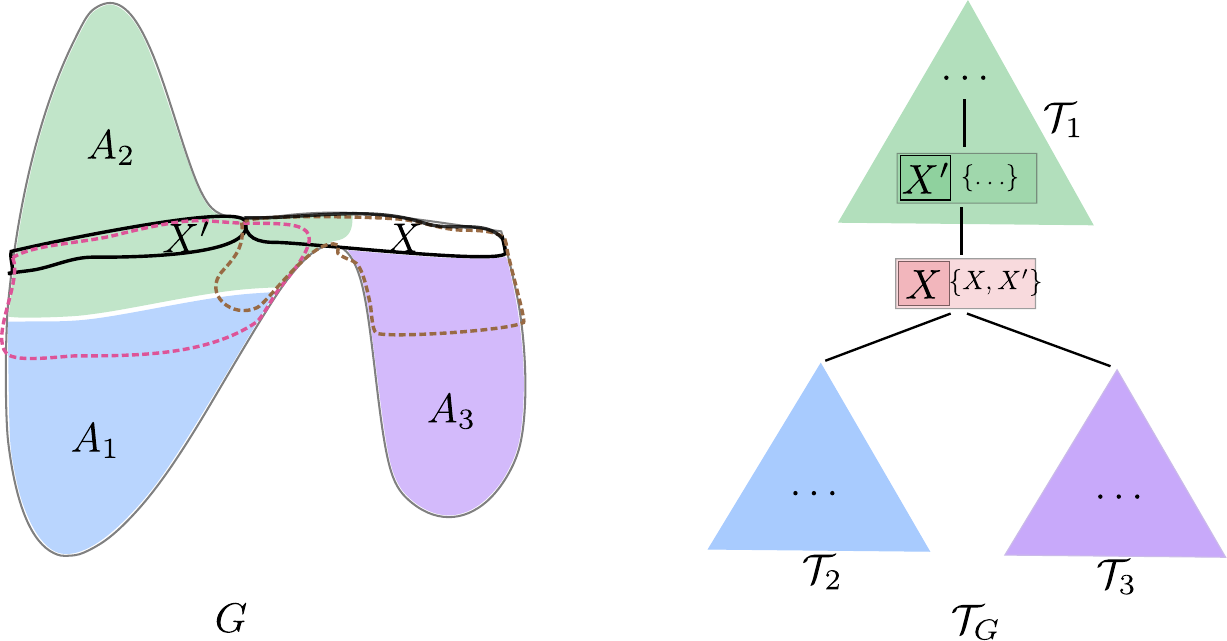}
    \caption{\textit{A stylized depiction of the divide-and-conquer algorithm for sparse cover. Initially all vertices are active. After deleting supernode $X$ from $\cT$, we recuse on each connected component $\cT_i$ of $\cT \setminus X$ (shown right) and its associated active vertex set $A_i$ (shown left).}}
    \label{fig:recursive-cover}
\end{figure}

To prove correctness, we consider an arbitrary ball $B \coloneqq B(v,\Delta)$ whose vertices are all active ($B\subseteq A$), and show that $B$ is fully contained in some cluster.
If $B$ is fully contained in some subset of active vertices $A_i$, then the recursive call to $(\cT_i,A_i)$ will insure that $B$ is fully contained in some cluster of $\cC_i$.
Otherwise, there will be vertices $x,y\in B$ such that the path from $\eta[x]$ to $\eta[y]$ in $\cT$ goes though the centroid supernode $X$.
Let $\EMPH{$v_1$}\in B$ be the vertex such that $\eta[v_1]$ is of minimum depth (among all supernodes $\{\eta[u]\}_{u\in B}$, with respect to $\cT$).
As $\cT$ can be viewed also as a tree decomposition of the graph, $\eta[v_1]$ is an ancestor of all the supernodes $\{\eta[u]\}_{u\in B}$, and in particular, $B\subseteq \dom(\eta[v_1])$.
From the properties of tree decomposition, the shortest path from $\eta[v_1]$ to $v_1$ must go through some supernode $X'\in\bag{X}$. It follows that the cluster $C_{ X'}=B_{\dom( X')}( X',4\Delta)\cap A$ contains the ball $B$, as required.
Finally as the depth of the recursion is at most $\log n$ (as each connected component of $\cT\setminus X$ has at most $|\cT|/2$ supernodes),
and in each level of the recursion we create at most $r-1$ clusters,
the sparsity of the created cover is at most $r\cdot\log n$.
While the diameter of each cluster may be large, the existence of the skeleton guarantees the existence of a $(\Omega(1), O(r), O(\Delta))$ sparse cover for the cluster.
This leads to a sparse cover scheme for $G$ with padding $O(1)$ and sparsity $\poly r \cdot O(\Delta)$;
in particular, due to \Cref{thm:coversToPaddedDecomp}, a padding decomposition with padding parameter $O(\log r+\log\log n)$ follows.

\paragraph*{Sparse Covers Using Separator Supernodes.}
To obtain the desired $O(\log r)$ padding parameter we need a sparse cover with sparsity $\poly(r)$.
Our divide and conquer algorithm above used a centroid supernode of a cop decomposition, and thus divides the number of supernodes in each step of the recursion by $2$, leading to recursion depth $O(\log n)$.
In order to avoid this dependence,
we use the \emph{buffered cop decomposition} of \cite{CCLMST24} combined with techniques inspired by the recent padded decomposition for bounded-treewidth graphs \cite{FFIKLMZ24}. Rather than removing a single centroid $X$, we remove simultaneously a set of \EMPH{separator supernodes $\cX$} from $\cT$. We show how to construct a set of separator supernodes such that 
(1) every vertex in $V(G)$ is ``threatened'' by $O(r)$ separator supernodes; and
(2) for every remaining connected component $\cT_i$ in $\cT\setminus\cX$, and for every supernode $\eta\in\cT_i$,
we have $|\bag{\eta} \cap \cT_i| \le |\bag{\eta} \cap \cT| - 1$ (that is, the size of the bag of $\eta$ ``with respect to $\cT$'' is reduced).
In particular, the depth of the recursion will be bounded by $r$.

We choose the set $\cX$ as follows.
The root supernode $X_{\rm rt}$ of $\cT$ joins $\cX$. 
All the supernodes $\eta$ such that $X_{\rm rt}\in\bag{\eta}$ become marked and will not join $\cX$.
In general, we pick an unmarked supernode $X$ of minimum depth and add it to $\cX$.
All the supernodes $\eta$ such that $\bag \eta$ contains any supernode in $\bag{X} \cap \cT$ becomes marked and will not join $\cX$. 
We continue with the process until all the supernodes in $\cT$ are marked. 
See \Cref{fig:sep-supernodes} for an illustration.
The buffer property of buffered cop decomposition can be used to show that every vertex in $G$ is threatened by few supernodes in $\cX$.
Moreover, for every connected component $\cT_i$ of $\cT\setminus\cX$, and every supernode $\eta$ of $\cT_i$, there is some supernode in $\bag{\eta} \cap \cT$ that is separated from $\eta$ in $\cT$ by $\cX$.
Thus the size of $\bag{\eta}\cap\cT$ is reduced in every recursive call, and the depth of the recursion will be at most $r-1$.
To construct sparse cover, we use the same divide-and-conquer algorithm described before, but we divide based on the set $\cX$ instead of a single centroid node $X$. This leads to a sparse cover with $O(\poly r)$ sparsity, and thus to $O(\log r)$ padding parameter.

\section{Preliminaries}
$\tilde{O}$ notation hides poly-logarithmic factors, that is $\tilde{O}(g)=O(g)\cdot\polylog(g)$. All logarithms are at base $2$ (unless specified otherwise), $\ln$ stand for the natural logarithm.

We consider connected undirected weighted graphs $G=(V,E,w)$, where 
$w: E \to \R_{\ge 0}$ is a weight function.
We also use $V(G)$ and $E(G)$ to denote the vertices and edges of $G$, respectively.
We say that vertices $v,u\in V(G)$ are neighbors if $\{v,u\}\in E(G)$, and two clusters $C,C'\subseteq V(G)$ are neighbors if there are vertices $v\in C$, $u\in C'$, such that $v$ and $u$ are neighbors. 
Given a path $P=(v_0,v_1,\dots,v_k)$, $\|P\|=\sum_{i=1}^{k}w(\{v_{i-1},v_i\})$ denotes its length.
Let $\dist_{G}$ denote the shortest path metric in $G$, that is $\dist_G(u,v)=\min\left\{\|P\|\mid P\mbox{ is a path from }u\mbox{ to }v\right\}$.
Let $B_G(v,r)=\{u\in V(G)\mid d_G(v,u)\le r\}$ denote the closed ball of radius $r$ around $v$. 
For a vertex $v\in V(G)$ and a subset $A\subseteq V(G)$, let $d_{G}(x,A):=\min_{a\in A}d_G(x,a)$,
where $d_{G}(x,\varnothing)= \infty$. 
A ball around a set $B_G(A,r)$ is defined similarly.
For two subsets $A_1,A_2,\subseteq V(G)$, their distance is defined to be 
$d_{G}(A_1,A_2):=\min_{x\in A_1,y\in A_2}d_G(x,y)$.
For a subset of vertices
$A\subseteq V(G)$, $G[A]$ denotes the induced graph on $A$,
and $G\setminus A := G[V(G)\setminus A]$.

A graph $H$ is a \EMPH{minor} of a graph $G$ if we can obtain $H$ from
$G$ by edge deletions/contractions, and isolated vertex deletions.  A graph
family $\mathcal{G}$ is \EMPH{$H$-minor-free} if no graph
$G\in\mathcal{G}$ has $H$ as a minor.
Some examples of minor free graphs are planar graphs ($K_5$ and $K_{3,3}$ minor-free), outer-planar graphs ($K_4$ and $K_{3,2}$ minor-free), series-parallel graphs ($K_4$ minor-free) and trees ($K_3$ minor-free).

A \EMPH{tree decomposition} of a graph $G$ is a  tree $\mathcal{T}$ whose nodes are 
subsets $S$ of $V(G)$ called \emph{bags}, such that: (i) $\cup_{S\in V(\mathcal{T})} S = V(G)$, (ii) for every edge $(u,v) \in E(G)$, there exists a bag $S$ in $V(\cT)$ such that $\{u,v\}\subseteq S$, and (iii) for every $u\in V(G)$, the bags containing $u$ induces a connected subtree of $\mathcal{T}$. The \emph{width} of $\mathcal{T}$ is $\max_{S\in V(\mathcal{T})}\{|S|\}$-1. The \emph{treewidth} of $G$ is the minimum width among all possible tree decompositions of $G$.

\subsection{Buffered cop decomposition.}
The buffered cop decomposition was introduced by \cite{CCLMST24}, building on the work of \cite{And86, AGGNT19}. We recall their definition (see \Cref{fig:CopDecomp} for an illustration).
\begin{definition}
    A \EMPH{supernode} $\eta$ with \EMPH{skeleton} $T_\eta$ and \EMPH{radius} $\Delta$ is set of vertices $\eta \subseteq V(G)$ and a tree $T_\eta$, such that (i) $T_\eta$ is contained in $\eta$, and (ii) every vertex $\eta$ is within distance $\Delta$ from $T_\eta$, where distance is measured w.r.t. shortest-paths in $G[\eta]$.
\end{definition}
\begin{definition}
    A \EMPH{partition tree} for $G$ is a rooted tree $\cT_G$ whose vertices are supernodes of $G$, such that $V(\cT_G)$ is a partition of $V(G)$.
    For any supernode $\eta \in V(\cT_G)$, we define the \EMPH{domain $\dom(\eta)$} to be the subgraph of $G$ induced by the union of all vertices in supernodes in the subtree of $\cT_G$ rooted at $\eta$.
    We define \EMPH{$\bag \eta$} to be the set containing $\eta$ and all ancestor supernodes $\eta'$ of $\eta$ such that there is an edge between $\eta'$ and $\dom(\eta)$ in $G$. 
\end{definition}

\begin{definition}
\label{def:buffered-cop}
    A \EMPH{$(\Delta, \gamma, w)$-buffered cop decomposition} for $G$ is a partition tree $\cT_G$ for $G$, such that every supernode $\eta \in V(\cT_G)$ satisfies the following properties:
    \begin{itemize}
        \item \textnormal{[Supernode radius.]} The supernode $\eta$ has radius at most $\Delta$.
        
        \item \textnormal{[Shortest-path skeleton.]} The skeleton $T_\eta$ of $\eta$ is an SSSP tree in $\dom(\eta)$ with at most $w-1$ leaves (not counting the root).
        
        \item \textnormal{[Supernode buffer.]} Let $\eta'$ be another supernode that is an ancestor of $\eta$. Then either $\eta'\in\bag\eta$ (i.e. $\dom(\eta)$ and $\eta'$ are adjacent in $G$), or for every vertex $v$ in $\dom(\eta)$, we have $\dist_{\dom(\eta')}(v, \eta') \ge \gamma$.%
        \item \textnormal{[Tree decomposition.]}
        $\lvert \bag \eta \rvert \le w$.
        Further, define $\hat B_\eta \subseteq V(G)$ to be the set of vertices contained in some supernode in $B_\eta$, that is, $\hat B_\eta = \bigcup_{X \in \bag \eta} X$; and define the \EMPH{expansion} of $\cT_G$, denoted \EMPH{$\hat \cT_G$}, to be a tree isomorphic to $\cT_G$ with vertex set $\set{\hat B_\eta}_{\eta \in V(\cT_G)}$. Then $\hat \cT_G$ is a tree decomposition of $G$.

    \end{itemize}
\end{definition}

\begin{lemma}[Theorem 3.15 of \cite{CCLMST24}]
\label{lem:cclmst-cop}
Let $G$ be a $K_r$-minor-free graph, and let $\Delta$ be a positive number. Then $G$ admits a $(\Delta, \Delta/r, r-1)$-buffered cop decomposition.   
\end{lemma}

For a given $(\Delta, \gamma, w)$-buffered cop decomposition $\cT_G$, for any supernode $\eta$ in $V(\cT_G)$, we let \EMPH{$\net \eta$} denote some (arbitrary) \EMPH{$\Delta$-net} of the skeleton $T_\eta$.
That is, $\net \eta$ is a subset of vertices from $T_\eta$ such that the distance between every two net points $u,v\in\net{\eta}$ satisfies $\dist_{G[T_\eta]}(u,v)>\Delta$, and for every vertex $u\in T_\eta$ there is a net point $v\in \net{\eta}$ such that $\dist_{G[T_\eta]}(u,v)\le\Delta$. 
The following observation follows almost immediately from [shortest-path skeleton] property (see e.g. Claim 1 of \cite{Fil24sparse}).
\begin{observation}
\label{obs:net}
    Let $\cT$ be a $(\Delta, \gamma, w)$-buffered cop decomposition, and let $\eta$ be a supernode in $\cT$. For any number $\alpha \ge 1$ and any vertex $v\in\dom(\eta)$, there are at most $O(\alpha \cdot w)$ net points $p \in \net \eta$ such that $\dist_{\dom(\eta)}(v, p) \le \alpha \cdot \Delta$.
\end{observation}

\section{From buffered cop decomposition to sparse cover}
This section is devoted to proving the following theorem.
\begin{theorem}
\label{thm:BufferedToSparseCover}
    Suppose that graph $G$ has a $(\Delta, \gamma, w)$-buffered cop decomposition. Then for every $\rho \ge 1$,
    graph $G$ admits an
    $\left(8+\frac4\rho, O(\rho^2\cdot\frac{\Delta}{\gamma} \cdot w^3), (4+8\rho)\Delta\right)$-sparse partition cover.
\end{theorem}
The sparse partition cover of \Cref{thm:BufferedToSparseCover} guarantees that every ball of radius \EMPH{$\rho\Delta$} is contained in some cluster of the sparse partition cover.
For the sake of simplicity, one can choose $\rho =1$ to construct a $(12, O(w^3 \cdot \Delta/\gamma), 12\Delta)$-sparse partition cover. To optimize the padding parameter, one can choose $\rho = 4/\e$, and thereby get a $(8+\e, O(\e^2 \cdot w^3 \cdot \Delta/\gamma), O(\Delta/\e))$-sparse partition cover.
Given a $K_r$-minor-free graph $G$ and $\Delta>0$, by \Cref{lem:cclmst-cop}, $G$ admits a $(\Delta, \Delta/r, r-1)$-buffered cop decomposition. 
By fixing $\rho=4/\eps$, \Cref{thm:BufferedToSparseCover} implies that $G$ admits a $(8+\eps, O(r^4/\eps^2))$-\SPCS, as claimed in \Cref{thm:cover}.

The proof of \Cref{thm:BufferedToSparseCover} is divided into two subsection. 
In \Cref{subsec:SeparatorSupernodes} we present the procedure \textsc{SeparatorSupernodes}(\cT) which find a collection of separator supernodes $\cX$ in $\cT$ such that every vertex is ``close'' only to a small number of supernodes, while the removal of this supernodes breaks $\cT$ into connected components such that the width of the associated tree decomposition is reduced (see \Cref{lem:good-separator-nodes}).
Then, in \Cref{subsec:SparseCoverConstruction} we construct a sparse cover in the procedure \textsc{Cover}. This is done by first creating clusters from the separator supernodes, and then recursing on each connected component in $\cX$.

For the rest of this section, let \EMPH{$\Delta$}, \EMPH{$\gamma$}, and \EMPH{$w$} be fixed parameters, let \EMPH{$G$} be a graph, and let \EMPH{$\cT_G$} be a $(\Delta, \gamma, w)$-buffered cop decomposition of $G$.
For any subtree $\cT$ of $\cT_G$, we use \EMPH{$V(\cT)$} to denote the set of supernodes which are nodes of $\cT$.
We define the root of $\cT$, denoted \EMPH{$\Root(\cT)$}, to be the supernode of $\cT$ of minimum depth in $\cT_G$ (i.e., the node closest in $\cT_G$ to the root of $\cT_G$).

\subsection{Selecting separator supernodes in a cop decomposition}\label{subsec:SeparatorSupernodes}
In this section, we introduce a procedure \EMPH{$\textsc{SeparatorSupernodes}(\cT)$}; it takes as input a subtree $\cT$ of $\cT_G$, and it outputs a subset $\cX\subseteq V(\cT)$ of the supernodes in $\cT$ which we call the \EMPH{separator supernodes}.

\begin{definition}
    Let $\cT$ be a subtree of $\cT_G$. We say that $\cT$ has \EMPH{subtree-width $\hat w$} if, for every supernode $\eta$ in $V(\cT)$, the size of $\bag{\eta} \cap V(\cT)$ is at most $\hat w$.
\end{definition}
    Observe that the $(\Delta, \gamma, w)$-buffered cop decomposition $\cT_G$ has subtree-width $w$, and only the empty subtree $\varnothing$ has subtree-width 0. 
    The following lemma will be used for the ``divide'' part in our divide-and-conquer algorithm.
\begin{lemma}
\label{lem:good-separator-nodes}
    Let $\cT$ be a subtree of the $(\Delta, \gamma, w)$-buffered cop decomposition $\cT_G$. There is a procedure $\textsc{SeparatorSupernodes}(\cT)$ that returns a set of supernodes $\cX \subseteq V(\cT)$ such that:
    \begin{itemize}
        \item \textnormal{[Bounded threateners.]} For any vertex $v$ in the graph $G$ and $\alpha \ge 1$, there are $O(\alpha \cdot \Delta/\gamma)$ separator supernodes $X \in \cX$ such that (i) $v \in \dom(X)$ and (ii) there exists some $X' \in \bag X \cap V(\cT)$ with $\dist_{\dom(X')}(v, X') \le \alpha \cdot \Delta$.
        \item \textnormal{[Bounded recursion.]} If $\cT$ has subtree-width $\hat w$, then each subtree in \EMPH{$\cT \setminus \cX$} (that is, the set of connected components obtained by deleting $\cX$ vertices from $\cT$) has subtree-width at most $\hat w - 1$.
    \end{itemize}
\end{lemma}

In the next section, we use the [bounded recursion] property to show that our (recursive) sparse cover construction terminates after $O(w)$ iterations; we use the [bounded threateners] property to show that each iteration adds only $O(w\cdot\hat{w} \cdot \Delta/\gamma)$ overlapping clusters in the sparse cover. The rest of this subsection is dedicated to proving \Cref{lem:good-separator-nodes}.

\begin{tcolorbox}
$\textsc{SeparatorSupernodes}(\cT)$:
\begin{enumerate}
    \item Initialize $\cX \gets \varnothing$ to be a set of separator supernodes.
    Initially, every node in $\cT$ is \EMPH{unmarked}.
    \item While there is an unmarked node of $\cT$, do the following: Let $X$ be an unmarked node in $\cT$ if minimum depth (i.e., a node that is closest in tree $\cT$ to the root of $\cT$), and add $X$ to $\cX$. Then, letting $\mathcal{M} \gets \bag X \cap V(\cT)$, \EMPH{mark} all nodes $\eta$ in $\cT$ such that (i) $\eta$ is a descendant of $X$ in $\cT$, and (ii) $\bag \eta$ contains some supernode in $\mathcal{M}$. In particular, note that $X$ is marked.
    \item Once every node of $\cT$ is marked, return $\cX$.
\end{enumerate}
\end{tcolorbox}

\begin{figure}[h!]
    \centering
    \includegraphics[width=\linewidth]{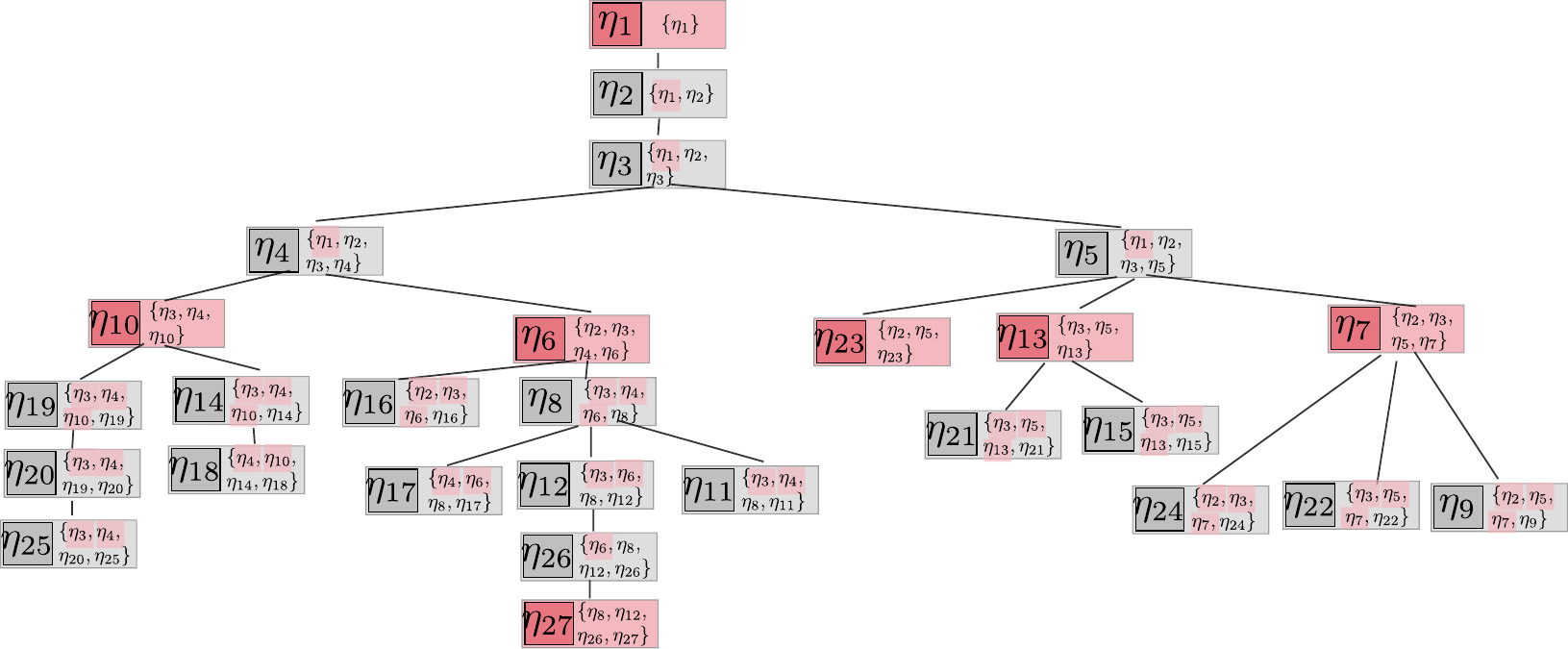}
    \caption{
    \textit{Illustration of the partition tree $\cT_G$ (of the buffered cop decomposition from \Cref{fig:CopDecomp}), and the output of $\textsc{SeparatorSupernodes}(\cT_G)$.    
    Supernodes in $\cX$ are red; supernodes not in $\cX$ are gray. 
    For each supernode $\eta$ that is \emph{not} in $\cX$, we highlight (in pink) each supernode in $\bag{\eta}$ that appears in the bag of some ancestor supernode in $\cX$; these supernodes witness the fact that $\eta$ is ``marked'' and does not join $\cX$.
    }}
    \label{fig:sep-supernodes}
\end{figure}

Observe that $\textsc{SeparatorSupernodes}(\cT)$ terminates after finitely many iterations of Step 2, as each iteration causes at least one unmarked node (namely, $X$) to become marked.
In this section, we use two straightforward consequences of the [tree decomposition] property of buffered cop decomposition, which we state here for clarity (\Cref{obs:tree-decomp}). We then prove that $\textsc{SeparatorSupernodes}$ satisfies the [bounded threateners] property (\Cref{clm:sep-nodes-sparse}) and the [bounded recursion] property (\Cref{clm:subtree-width}); our \Cref{lem:good-separator-nodes} follows directly from these two claims.
\begin{observation}
\label{obs:tree-decomp}
    Let $\eta^{+}$, $\eta$, and $\eta^{-}$ be supernodes in $\cT$, where $\eta^{+}$ is not a descendant of $\eta$, and $\eta^{-}$ is a descendant of $\eta$.
    \begin{enumerate}
        \item[(i)] If $\eta^+ \in \bag{\eta^-}$, then $\eta^+ \in \bag{\eta}$.
        \item[(ii)] Let $v^{+}$ and $v^{-}$ be vertices with $v^{+} \in \eta^{+}$ and $v^{-} \in \eta^{-}$. Then any path $P$ in $G$ between $v^{+}$ and $v^{-}$ includes some vertex in a supernode $\eta' \in \bag {\eta}$, where $\eta' \neq \eta$.
    \end{enumerate}
\end{observation}
\begin{proof}
    \textbf{(i.)} This follows immediately from the definition of tree decomposition.

    \medskip \noindent
    \textbf{(ii.)} Viewing $P$ as a path starting at $v^+$ and ending at $v^-$, let $(a,b)$ be the first edge on $P$ such that $a \not \in \dom(\eta)$ and $b \in \dom(\eta)$.
    Let $\eta_a$ and $\eta_b$ be the supernodes containing $a$ and $b$, respectively.
    We choose $\eta'$ to be $\eta_a$.    
    By definition, $\eta_a$ is not a descendant of $\eta$ (and $\eta_a \neq \eta$).
    Now, by [tree decomposition] property, there is some supernode $\eta''$ such that $\bag{\eta''}$ contains both $\eta_a$ and $\eta_b$.
    In particular, $\eta''$ is a descendant of $\eta_b$, and thus $\eta''$ is also a descendant of $\eta$.
    Thus item (i) implies that $\eta_a \in \bag{\eta}$.    
\end{proof}

For the rest of this section, let \EMPH{$\cT$} be a subtree of $\cT_G$, and let \EMPH{$\cX$} be the output of $\textsc{SeparatorSupernodes}(\cT)$.

\begin{claim}
\label{clm:sep-nodes-sparse}
    For any vertex $v$ in $G$, let $\EMPH{$\cX[v]$} \subseteq \cX$ denote the set of separator supernodes $X \in \cX$ such that (i) $v \in \dom(X)$ and (ii) there exists some $X' \in \bag X \cap V(\cT)$ with $\dist_{\dom(X')}(v, X') \le \alpha \cdot \Delta$. Then $\cX[v]$ contains at most $2\alpha \cdot \Delta/\gamma + 2$ supernodes.
\end{claim}
\begin{proof}
Let \EMPH{$\eta$} denote the supdenode in $\cT_G$ that contains $v$.
Every supernode $X \in \cX[v]$ has $v \in \dom(X)$ and is therefore an ancestor of $\eta$, and so these supernodes can be put in a linear order based on ancestor-descendant relationship. 
Let \EMPH{$(X_0, \ldots, X_k)$} denote an ordered list of the supernodes in $\cX[v]$, where $X_i$ is a proper descendant of $X_{i+1}$ (note that the supernode $\eta$ is a descendant of $X_0$, where possibly $\eta = X_0$).
To prove the lemma, we must show that $k \le 2\alpha \cdot \Delta/\gamma + 1$. To this end, we prove the following claim by induction on $i$:
\begin{quote}
For every $i \in \{0,\dots,k\}$ and every $X_i' \in \bag{X_i} \cap V(\cT)$,
we have $\dist_{\dom(X_i')}(X_i', v) \ge \left\lfloor \frac{i}{2} \right\rfloor \cdot \gamma$ .
\end{quote}
This claim suffices to prove the lemma: as $X_k\in \cX[v]$, there is some $X_k' \in \bag{X_k} \cap V(\cT)$ with
$\dist_{\dom(X_k')}(X_k', v) \le \alpha \cdot \Delta$. It follows that $\left\lfloor \frac{k}{2}\right\rfloor \cdot\gamma\le\alpha\cdot\Delta$
and so $k\le2\alpha\cdot\Delta/\gamma+1$, and in particular 
$\left|\cX[v]\right|=k+1\le 2\alpha\cdot\Delta/\gamma+2$.

In the base case, $i = 0$ or $i=1$, and the claim holds trivially (as distances are nonnegative).
For the inductive step, we have $i \ge 2$. Let \EMPH{$P$} be a shortest path in $\dom(X_i')$ between $X_i'$ and $v$. By \Cref{obs:tree-decomp}(ii), the path $P$ intersects some supernode in $\bag{X_{i-2}}$.
Let \EMPH{$X_{i-2}'$} denote the last supernode in $\bag{X_{i-2}}$ that $P$ intersects, where we view $P$ as a path that starts at $X_i'$ and ends at $v$.
Let  $P_{\rm pre}$ be the \emph{prefix} of $P$ that travels from $X_{i}'$ to $X_{i-2}'$, 
and $P_{\rm suf}$ be the \emph{suffix} of $P$ that travels from $X_{i-2}'$ to $v$.
Clearly $P_{\rm pre}$ is contained in $\dom(X_i')$ (as $P$ is a path in $\dom(X_i')$).
We argue that $P_{\rm suf}$ is contained in $\dom(X_{i-2}')$. 
Indeed, for the sake of contradiction suppose that $P_{\rm suf}$ goes through a supernode $\eta'$ which is not a descendant of $X_{i-2}'$ in $\cT$. In particular, $\eta'$ is not a descendant of $X_{i-2}$. By \Cref{obs:tree-decomp}(ii), $P_{\rm suf}$ goes though a supernode $X_{i-2}''\in\bag{X_{i}}$, a contradiction to the maximality of $X_{i-2}'$.
We conclude:
\[
\dist_{\dom(X_{i}')}(X_{i}',v)=\norm{P}\ge\norm{P_{{\rm pre}}}+\norm{P_{{\rm suf}}}\ge\dist_{\dom(X_{i}')}(X_{i}',X_{i-2}')+\dist_{\dom(X_{i-2}')}(X_{i-2}',v)~.
\]
We lower bound each term. First we claim that $X_{i-2}' \in V(\cT)$; we can then apply induction to conclude $\dist_{\dom(X_{i-2}')}(X_{i-2}', v) \ge \lfloor \frac{i-2}{2} \rfloor \cdot \gamma$. 
Indeed, $X_{i-2}$ is a supernode in $V(\cT)$, so (because $\cT$ is connected) every ancestor of $X_{i-2}$ is either in $V(\cT)$ or an ancestor of $\Root(\cT)$. In particular, $X_{i-2}'$ is in $\bag {X_{i-2}}$ and is an ancestor of $X_{i-2}$, so it is either in $V(\cT)$ or an ancestor of $\Root(\cT)$. 
But, because $X_i'$ is in $\cT$ and $P$ is a path in $\dom(X_i')$, $P$ does not intersect any supernode that is an ancestor of $\Root(\cT)$.
Thus $X_{i-2}' \in V(\cT)$.

Next, we show that $X_i' \not \in \bag{X_{i-2}'}$, and so the [supernode buffer] property of $\cT_G$ implies that $\dist_{\dom(X_i')}(X_i', X_{i-2}') \ge \gamma$; the claim then follows.
We first argue that  $X_{i-2}'$ is a descendant of $X_{i-1}$. 
For the sake of contradiction, suppose otherwise;
that is, $X_{i-2}'$ is not a descendant of $X_{i-1}$. As $X_{i-2}$ is a descendant of $X_{i-1}$, and $X_{i-2}' \in \bag{X_{i-2}}$, \Cref{obs:tree-decomp}(i) implies that $X_{i-2}' \in \bag{X_{i-1}}$.
This leads to a contradiction: the \textsc{SeparatorSupernodes} algorithm selects $X_{i-1}$ to join $\cX$ before $X_{i-2}$ (because $X_{i-1}$ is an ancestor of $X_{i-2}$), and if $X_{i-2}'$ were in both $\bag{X_{i-1}}$ and $\bag{X_{i-2}}$, then $X_{i-2}$ would have been ``marked'' at the time $X_{i-1}$ was added to $\cX$;  thus $X_{i-2}$ would not have been added to $\cX$ itself.

We finish the proof with a similar argument. Suppose for the sake of contradiction that $X_i' \in \bag{X_{i-2}'}$.
By \Cref{obs:tree-decomp}(i), as $X_{i-2}'$ is a descendant of $X_{i-1}$, and $X_{i-1}$ is a descendant of $X_{i}'$, it holds that $X_i' \in \bag{X_{i-1}}$.
This is a contradiction: the \textsc{SeparatorSupernodes} algorithm selects $X_i$ to join $\cX$ before $X_{i-1}$, and if $X_i'$ is in both $\bag{X_i}$ and $\bag{X_{i-1}}$ then $X_{i-1}$ would not have been selected to join $\cX$.
\end{proof}

\begin{claim}
\label{clm:subtree-width}
Let $\hat w > 0$. If $\cT$ has subtree-width $\hat w$, then every subtree in $\cT \setminus \cX$ has subtree-width at most $\hat w-1$.
\end{claim}
\begin{proof}
Let \EMPH{$\cT'$} be an arbitrary connected component in $\cT \setminus \cX$. Note that $\Root(\cT) \in \cX$ by construction, so
the parent of $\Root(\cT')$ exists and is some separator supernode \EMPH{$X$}$\in \cX$.
Consider a supernode $\EMPH{$\eta$}\in V(\cT')$. To prove the claim, we have to show that the bag size of $\eta$ with respect to $\cT'$ is bounded by $\hat{w}-1$; that is, $|\bag{\eta} \cap V(\cT')|\le \hat w -1$. 
As the subtree-width of $\cT$ is at most $\hat{w}$, it is enough to show that there is some supernode in $\bag{\eta} \cap V(\cT)$, i.e. the bag of $\eta$ with respect to $\cT$, which does not belong to $V(\cT')$.

The supernode $\eta$ was not selected as a separator supernode, so it must have been ``marked'' at some point when some ancestor \EMPH{$X_\eta$} of $\eta$ was selected as a separator supernode.
The fact $\eta$ is marked means there is some supernode \EMPH{$X_{\eta}'$} such that $X_{\eta}' \in \bag{X_\eta} \cap \bag{\eta} \cap V(\cT)$. As $\cT'$ is connected, and every supernode in $\cT'$ is not a separator supernode, $X_\eta$ is a (not necessarily proper) ancestor of $X$.\footnote{It is not hard to show that $X_\eta = X$, but we do not need this fact.} \Cref{obs:tree-decomp}(i) implies that $X_{\eta}' \in \bag{X}$.
In particular, $X_{\eta}'$ is ancestor of $X$ and does not belong to $\cT'$, as required.
\end{proof}

\subsection{Sparse cover construction}\label{subsec:SparseCoverConstruction}
The procedure $\textsc{Cover}(\cT,A)$ takes as an input a subtree $\cT$, and a subset of vertices $A$.
Each vertex in $A$ will be at distance at most $2\Delta$ from some supernode in $\cT$.
The procedure returns a sparse cover for the vertices in $A$.
Recall that for a supernode $\eta$, $\net{\eta}$ is a subset of points from $T_\eta$ at minimum pairwise distance at least $\Delta$, and such that every vertex in $T_\eta$ has a net point at distance at most $\Delta$. 
\begin{tcolorbox}
$\textsc{Cover}(\cT,A):$
\begin{enumerate}
    \item \emph{Select separator supernodes $\cX$.}

    If $\cT = \varnothing$, return $\varnothing$. Otherwise, let $\EMPH{$\cX$} \gets \textsc{SeparatorSupernodes}(\cT)$.

    \item \emph{Grow clusters around $\cX$.}

    For every supernode $X$ in $\cX$, for every supernode $X' \in \bag{X} \cap V(\cT)$, and for every net point $p\in\net{X'}$, create a cluster $\EMPH{$C_{X,p}$} \gets B_{\dom(X')}(p, (2+4\rho)\cdot\Delta) \cap \dom(X)\cap A$.
    
    Let \EMPH{$\cC$} denote the set of all such clusters.

    \item \emph{Recurse.}
    
    For every vertex $v \in A$, define \EMPH{$\eta[v]$} to be the highest supernode in $\cT$ such that $\dist_{\dom(\eta_v)}(v, \eta_v) \le 2\rho \Delta$.
    For every connected component \EMPH{$\cT_i$} $ \in \cT \setminus \cX$, define \EMPH{$A_i$} to be the set of all vertices $v \in A$ such that $\eta_v$ is in $\cT_i$.
    
    Recursively compute $\cC_i \gets \textsc{Cover}(\cT_i,A_i)$ for each $i$.
    
    Return the set of all clusters created: $\cC \cup \bigcup_i \cC_i$.
\end{enumerate}
\end{tcolorbox}

As the diameter of every ball is at most twice the radius, we have:
\begin{observation}[Diameter]
\label{obs:diameter}
    Let $\cT$ be a subtree of $\cT_G$ and let $\mathfrak{C} \gets \textsc{Cover}(\cT)$. Every cluster $C \in \mathfrak{C}$ has (weak) diameter at most $(4+8\rho)\cdot\Delta$.
\end{observation}

In the next lemma we bound the number of clusters containing a single vertex. Clearly the clusters of $\textsc{Cover}(\cT, A)$ only contain vertices in $A$.

\begin{lemma}[Sparsity]
\label{lem:sparsity}
    Let $\cT$ be a subtree of $\cT_G$ of width $\hat{w}$, and let $\mathfrak{C} \gets \textsc{Cover}(\cT,A)$. 
    Then the clusters in $\mathfrak{C}$ can be partitioned into $O(\rho^{2}\cdot\frac{\Delta}{\gamma} \cdot w \cdot \hat w^2)$ partitions $\cP_1,\cP_2,\dots$ such that for every $i$, all the clusters in $\cP_i$ are disjoint.
\end{lemma}
\begin{proof}    
    Let $\kappa$ be some large enough constant (to be chosen later). We argue by induction on $\hat w$ that 
    the clusters in $\mathfrak{C}$ can be partitioned into  $\kappa\cdot\rho^{2}\cdot\frac{\Delta}{\gamma} \cdot w \cdot \hat w^2$ partitions $\cP_1,\cP_2,\dots$ such that for every $i$, all the clusters in $\cP_i$ are disjoint. 
    In the base case, $\hat w = 0$ and $\cT = \varnothing$, so we must have $A = \varnothing$. 
    For the inductive case, consider $\hat{w} > 0$.
    We first partition the clusters in $\cC$. 
    Let $\EMPH{$\widetilde{\cX}$}=\bigcup_{X\in\cX}\bag{X}$ be all the supernodes from which we grew clusters. 
    Every cluster $C_{X,p}\in\cC$ is created from a supernode $X\in\widetilde{\cX}$, and a net point $p\in\net{X}$, where $C_{X,p}=B_{\dom(X)}(p,(2+4\rho)\cdot\Delta)$.
    Fix some $\EMPH{Y}\in\widetilde{\cX}$, and $\EMPH{$q$}\in\net{Y}$. Let \EMPH{$\Gamma(C_{Y,q})$} be all the clusters $C_{X,p}\in\cC$ such that
    $q \in \dom(X)$ and $C_{Y,q}$ intersects $C_{X,p}$.
    Note that, for any supernode $X$ with $q \in \dom(X)$ and any $p\in\net{X}$, if $C_{Y,q}$ and $C_{X,p}$ intersect, then it must hold that $d_{\dom(X)}(p,q)\le 4+8\rho$.
    We 
    count the number of clusters in $\Gamma(C_{Y,q})$:
    \begin{itemize}
        \item By the [bounded threateners] property of \Cref{lem:good-separator-nodes}, there are at most $O(\frac{\rho\cdot\Delta}{\gamma})$ supernodes \EMPH{$X$} in $\cX$ such that $q \in \dom(X)$ and there is some $X' \in \bag{X} \cap V(\cT)$ with $\dist_{\dom(X')}(X', \dom(Y)) \le (4+8\rho)\cdot\Delta.$
        \item For each such supernode $X$, there are at most $\hat{w}$ supernodes $\EMPH{$X'$} \in \bag{X} \cap V(\cT)$, as the subtree-width of $\cT$ is $\hat{w}$.
        \item For each such supernode $X'$, \Cref{obs:net} implies that there are $O(\rho\cdot w)$ net points $\EMPH{$p$}\in\net{X'}$ such that $\dist_{\dom(X')}(p, q) \le (4+8\rho)\Delta$.
    \end{itemize}    
    We conclude that  $\left|\Gamma(C_{Y,q})\right|=O(\frac{\rho\cdot\Delta}{\gamma}\cdot\rho w\cdot\hat{w})\le\kappa\cdot\rho^{2}\cdot\frac{\Delta}{\gamma}\cdot w\cdot\hat{w}$, for a sufficiently large $\kappa$.

    The clusters $\cC$ are partitioned in a greedy manner:
    Initially all the partitions sets $\cP_1,\cP_2,\dots$ are empty. Then, we go over the clusters $C_{X,p}\in\cC$ in non-decreasing order of depth (of $X$, w.r.t. $\cT$), and arbitrary order of $p\in\net{X}$.
    The cluster $C_{X,p}$ will join the partition $\cP_i$ with the minimum index $i$ such that there is no cluster in $\cP_i$ intersecting $C_{X,p}$.
    As among the clusters intersecting $C_{X,p}$ only the clusters in $\Gamma(C_{X,p})$ could be examined before $C_{X,p}$, it holds that $C_{X,P}$ joins partition $\cP_i$ for $i\le|\Gamma(C_{X,p})|$ (note that $C_{X,p}\in \Gamma(C_{X,p})$).
    Thus we partitioned the clusters of $\cC$ into at most $\kappa\cdot\rho^{2}\cdot\frac{\Delta}{\gamma}\cdot w\cdot\hat{w}$ partitions.

    Next, consider a connected component $\cT_j\in\cT\setminus\cX$, with a subset $A_j$ of active vertices.
    By the [bounded recursion] property of \Cref{lem:good-separator-nodes}, the corresponding tree $\cT_j$ has subtree-width at most $\hat w - 1$. 
    Thus, using the induction hypothesis, the clusters in $\cC_j=\textsc{Cover}(\cT_{j},A_{j})$ can be partitioned into at most $\kappa\cdot \rho^{2}\cdot\frac{\Delta}{\gamma} \cdot w \cdot  (\hat w-1)^2$ partitions $\cP^{(j)}_1,\cP^{(j)}_2,\dots$ such that all the clusters in every partition $\cP^{(j)}_l$ are disjoint.
    Given two connected components $\cT_{j_1},\cT_{j_2}\in\cT\setminus\cX$, the corresponding active sets $A_{j_1},A_{j_2}$ are disjoint, and so there are no two clusters $C_1\in\cC_{j_1}$, $C_2\in\cC_{j_2}$ that intersect.
    It follows that we can arbitrarily combine such partitions (for the different connected components).
    In total, the clusters in $\mathfrak{C}$ can be partitioned into at most $\kappa\cdot\rho^{2}\cdot\frac{\Delta}{\gamma}\cdot w\cdot\left((\hat{w}-1)^{2}+\hat{w}\right)\le\kappa\cdot\rho^{2}\cdot\frac{\Delta}{\gamma}\cdot w\cdot\hat{w}^{2}$ partitions, such that every two clusters in the same partition are disjoint, proving the claim.
\end{proof}

It remains to prove that every $\rho \Delta$-ball is contained in some cluster returned by $\textsc{Cover}$. 
The following technical claim will be used several times during the proof of this property.

\begin{claim}
\label{clm:close-ancestor}
    Let $\cT$ be a subtree of $\cT_G$, let $\eta$ be a supernode in $V(\cT)$, and let $u$ and $v$ be vertices in $\dom(\Root(\cT))$ with $u \not \in \dom(\eta)$ and $v \in \dom(\eta)$. Then there is some supernode $\eta'\in\bag \eta \cap V(\cT)$ such that $\eta' \neq \eta$, and $\dist_{\dom(\eta')}(\eta', v) \le \dist_{\dom(\Root(\cT))}(u,v)$.
    
\end{claim}

\begin{wrapfigure}{r}{0.12\textwidth}
\centering
		\vspace{-10pt}
  \includegraphics[width=0.11\textwidth]{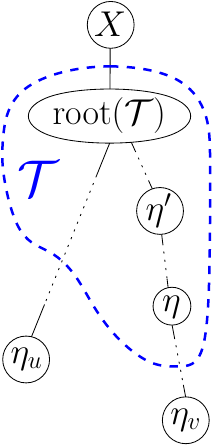}
\end{wrapfigure}

\noindent \textbf{Proof:}
Let \EMPH{$P$} be a shortest path in $\dom(\Root(\cT))$ between $u$ and $v$. By definition of $\dom(\cdot)$, vertex $u$ belongs to a supernode that is not a descendant of $\eta$, and $v$ belongs to a supernode that is a descendant of $\eta$. 
In the illustration on the right, $\cT$ is a subtree of $\cT_G$; and $u$ and $v$ belong respectively to the supernodes $\eta_u$ and $\eta_v$, which are descendants of $\Root(\cT)$ but are not in $V(\cT)$.
By \Cref{obs:tree-decomp}(ii), path $P$ intersects some supernode in $\bag{\eta}$ which is not $\eta$ itself.
Let $\EMPH{$\eta'$} \in \bag{\eta}$ be the last such supernode that $P$ intersects (when viewing $P$ as a path starting at $u$ and ending at $v$). We have $\eta' \neq \eta$ by assumption.
Let \EMPH{$P[\eta':v]$} be the suffix subpath of $P$ that starts at the last vertex in $\eta' \cap P$ and ends at $v$. By the choice of $\eta'$, the subpath $P[\eta':v]$ is contained in $\eta' \cup \dom(\eta)$; as $\eta'$ is an ancestor of $\eta$, this implies $P[\eta':v]$ is contained in $\dom(\eta')$. We conclude
    \[\dist_{\dom(\eta')}(\eta', v) \le \norm{P[\eta':v]} \le \norm{P} = \dist_{\dom(\Root(\cT))}(u,v)~.\]
    Finally, note that $\eta'$ is an ancestor of $\eta$, and descendant of $\Root(\cT)$. As $\cT$ is connected and both $\eta,\Root(\cT)$ belong to the subtree $\cT$, it follows that $\eta'\in V(\cT)$, as required. 
    \markatright\QED

\begin{lemma}[Covering]
\label{lem:covering}
    The procedure $\textsc{Cover}(\cT_G,V)$ returns a clustering $\mathfrak{C}$ such that every $\rho\Delta$ radius ball is contained in some cluster of $\mathfrak{C}$.
\end{lemma}
\begin{proof}
    The procedure $\textsc{Cover}(\cT,A)$ receives a subtree $\cT$ of $\cT_G$, and a subset $A\subseteq V$ of vertices.
    We think of the set $A$ as the set of active vertices. Note that a vertex $v\in A$ might not belong to any of the supernodes in $\cT$.
    Nevertheless, by a simple induction it follows that for every vertex $v\in A$, there is some supernode $\eta\in\cT$ such that $\dist_{\dom(\eta)}(v,\eta) \le 2\rho \Delta$.
    We prove the lemma by induction on the size of $\cT$. 

\begin{quote}
    \textbf{Induction Hypothesis:} Let $\textsc{Cover}(\cT,A)$ be a call made during the recursive execution of $\textsc{Cover}(\cT_G,V)$ which returned the cover $\mathfrak{C}$. Let $v_0 \in A$ be a vertex such that the ball $B \coloneqq B_{G}(v_0, \rho\Delta)$ is contained in $A$. Then there is some cluster in $\mathfrak{C}$ that contains $B$.    
\end{quote}
    
    The base case is when $\cT=\varnothing$, here also $A=\varnothing$ and there is nothing to prove (this case is actually never implemented by the algorithm and is useful only for the proof).
    For the inductive step, consider a ball $B=B_{G}(v_0, \rho\Delta)$ such that $B\subseteq A$.
    Suppose first that there is some connected component $\cT_i\in\cT\setminus\cX$ such that $B\subseteq A_i$, as defined in Step 3 of \textsc{Cover}. 
    By induction, the ball $B$ is contained in some cluster in $\cC_i$, and we are done. 
    In the remainder of the proof we will thus assume that the ball $B$ is not contained in any such set $A_i$.
    
    Recall that for any vertex $v$ in $A$, we use $\eta[v]$ to denote the supernode of minimum depth in $\cT$ such that $\dist_{\dom(\eta[v])}(v, \eta[v]) \le 2\rho \Delta$, and we assign $v$ to the set $A_i$ associated with the connected component $\cT_i\in\cT \setminus \cX$ containing $\eta[v]$. 
    Let \EMPH{$v_1$} be a vertex in $B$ that minimizes the depth of $\eta[v_1]$ (that is, there is no vertex $v\in B$ such that $\eta[v]$ is an ancestor of $\eta[v_1])$. 
    We first argue that for every other vertex $u\in B$, $\eta[v_1]$ is an ancestor of $\eta[u]$. Suppose for the sake of contradiction otherwise, and let $u\in B$ be such a vertex. 
    By the minimality of $v_1$, it follows that $\eta[v_1]$ and $\eta[u]$ are not in ancestor-descendant relation. Then $u \not \in \dom(\eta[v_1])$.     
    By triangle inequality, $B$ has strong diameter at most $2\rho\Delta$, hence $\dist_{\dom(\Root(\cT))}(v_1, u) \le 2\rho\Delta$. By \Cref{clm:close-ancestor}, there is some supernode $\eta'$ in $V(\cT)$ that is a proper ancestor of $\eta[v_1]$ with $\dist_{\dom(\eta')}(\eta', v_1) \le 2\rho\Delta$. This contradicts the definition of $\eta[v_1]$.
    We conclude that $\eta[v_1]$ is an ancestor of $\eta[u]$, for all $u\in B$.
    
    As $B$ is not contained in any set $A_i$, there exists some vertex $\EMPH{$v_2$}\in B$ such that $\eta[v_1]$ and $\eta[v_2]$ are not in the same connected component of $\cT \setminus \cX$. 
    Let $\EMPH{$X$} \in \cX$ be some separator supernode on the path in $\cT$ between $\eta[v_1]$ and $\eta[v_2]$ (note that it is possible that either $\eta[v_1]$ or $\eta[v_2]$ equal $X$).
    Such an $X$ exists because $\eta[v_1]$ and $\eta[v_2]$ are disconnected in $\cT \setminus \cX$.
    We next argue that $B$ is contained in $\dom(X)$. 
    For the sake of contradiction suppose otherwise. 
    Then there is some vertex $u \in B$ that is not in $\dom(X)$. 
    Using \Cref{clm:close-ancestor},
    there is some supernode $\eta'\in V(\cT)$ that is a proper ancestor of $X$ with $\dist_{\dom(\eta')}(\eta', v_2) \le\dist_{\dom(\Root(\cT))}(u, v_2) \le 2\rho\Delta$. But $X$ is an ancestor of $\eta[v_2]$, so $\eta'$ is a proper ancestor of $\eta[v_2]$, contradicting the definition of $\eta[v_2]$.
    We conclude that $B\subseteq\dom(X)$.

    Next, we claim that there is some supernode
    $\EMPH{$X_B$} \in \bag{X} \cap V(\cT)$ with $\dist_{\dom(X_B)}(X_B, v_1) \le 2\rho\Delta$.
    There are two cases. The first case is when $\eta[v_1] = X$; here we choose $X_B \coloneqq \eta[v_1]$. In the second case, $\eta[v_1]$ is a proper ancestor of $X$. By definition of $\eta[v_1]$, there is some $u \in \eta[v_1]$ such that $\dist_{\dom(\eta[v_1])}(u,v_1) \le 2\rho\Delta$; by \Cref{clm:close-ancestor}, there is some supernode $X_B\in \bag{X} \cap V(\cT)$ such that $\dist_{\dom(X_B)}(X_B, v_1) \le 2\rho\Delta$.

    We are now ready to find a cluster in $\cC$ that contains $B$. 
    Let $\EMPH{$v'$} \in X_B$ be a vertex such that $\dist_{\dom(X_B)}(v', v_1)= \dist_{\dom(X_B)}(X_B, v_1)\le 2\rho\Delta$.
    By the [supernode radius] property, there is some point \EMPH{$v''$} on the skeleton $T_{X_B}$ of $X_B$ with $\dist_{X_B}(v'', v') \le \Delta$.
    As $\net{X_B}$ is a $\Delta$-net of the skeleton of $X_B$, there is some point $\EMPH{$p$} \in \net{X_B}$ with $\dist_{X_B}(p, v'') \le \Delta$. 
    Using the triangle inequality (and the fact that $X_B$ is a subgraph of $\dom(X_B)$), for every vertex $u\in B$ it holds that
    \begin{align*}
    \dist_{\dom(X_{B})}(p,u) & \le\dist_{\dom(X_{B})}(p,v'')+\dist_{\dom(X_{B})}(v'',v')+\dist_{\dom(X_{B})}(v',v_{1})+\dist_{\dom(X_{B})}(v_{1},u)\\
     & \le\Delta+\Delta+2\rho\Delta+2\rho\Delta=(2+4\rho)\cdot\Delta~.
    \end{align*}
    It follows that the entire ball $B$ is contained in the ball $B_{\dom(X_B)}(p, (2+4\rho)\cdot\Delta)$. 
    As $B$ is contained in both $\dom(X)$ and $A$, it follows that $B$ is contained in the cluster $C_{X_B,p}$ created in Step 2 of $\textsc{Cover}(H, \cT)$, as required.
\end{proof}

By \Cref{obs:diameter} every cluster has diameter at most $(4+8\rho)\Delta$, while by \Cref{lem:covering} every ball of radius $\rho\Delta$ is fully contained in some cluster. It follows that our cover has padding $8+\frac4\rho$. By \Cref{lem:sparsity}, the clusters in $\mathfrak{C}$ can be partitions into at most  $O(\rho^2\cdot\frac{\Delta}{\gamma} \cdot w^3)$ partitions, as required. \Cref{thm:BufferedToSparseCover} now follows.

\section{From sparse covers to padded decompositions}

This section is devoted to proving our meta \Cref{thm:coversToPaddedDecomp} (restated below for convenience) that reduces sparse covers into padded decomposition.
\SparseCoverToDecomposition*
\begin{proof}
	Given a cover $\cC$, we will sample a partition $\cP$ with the desired padding properties. To avoid confusion, the sets of $\cC$ (resp. $\cP$) will be called $\cC$-clusters (resp. $\cP$-clusters).  
	Our decomposition is morally \cite{MPX13}-based, and closely follows \cite{Fil19padded}. Filtser \cite{Fil19padded} proved that the existence of a ``sparse net'' implies padded decomposition. Specifically, if there is a net $N$ such that every vertex $v$ has a net point at distance at most $d_G(v,N)\le \Delta$, and the number of net points in the ball $B_G(v,3\Delta)$ is at most $\tau$, then $G$ admits a strong $\left(O(\ln \tau),\frac{1}{16},4\Delta\right)$-padded decomposition. One can notice that by taking balls of radius $2\Delta$ around the sparse net $N$, we will get a $(\frac14,\tau,4\Delta)$-sparse cover, and thus our \Cref{thm:coversToPaddedDecomp} is in a sense generalization of \cite{Fil19padded} to arbitrary sparse covers (alas only with a weak diameter guarantee).
	
	For a $\cC$-cluster $C\subseteq V$ and a vertex $v\in V$, denote by $\EMPH{$\partial_{C}(v)$}\coloneqq d_G(v,V\setminus C)$ the distance between $v$ and the boundary of the $\cC$-cluster $C$.
	Note that if $v\notin C$, $\partial_{C}(v)=0$, while $B_G(v,r)\subseteq C$ implies  $\partial_{C}(v)>r$. Furthermore, for every pair of vertices $u,v\in V$, $|\partial_{C}(v)-\partial_{C}(u)|=\left|d_G(v,V\setminus C)-d_G(u,V\setminus C)\right|\le d_G(u,v)$.
	
	To create a padded decompositions, following previous works, we will use truncated exponential distribution. That is, exponential distribution conditioned on the event that the outcome lays in a certain interval. 
	The \emph{$[0,1]$-truncated exponential distribution} with
	parameter $\lambda$ is denoted by \EMPH{$\Texp(\lambda)$}, the density function is then
	$g(y) = \frac{ \lambda\cdot e^{-\lambda\cdot y} }{1 - e^{-\lambda}}$, for $y \in [0,1]$.

	\paragraph*{Construction.} Consider a $(\beta,s,\Delta)$-sparse cover $\cC$. For every $\cC$-cluster $C$, we sample $\EMPH{$\delta_C$}\sim\Texp(\lambda)$ using truncated exponential distribution with parameter $\EMPH{$\lambda$}=2+2\ln s$.
	For a $\cC$-cluster $C$, we define a function $\EMPH{$f_C$}:V\rightarrow\R_{\ge0}$, as follows: \[f_C(v)\coloneqq\delta_C\cdot\frac{\Delta}{\beta}+\partial_{C}(v).\]
	We create a partition $\EMPH{$\cP$}\coloneqq\{P_C\}_{C\in\cC}$ where each vertex $v\in V$ joins the $\cP$-cluster \EMPH{$P_C$} associated with the $\cC$-cluster $C$ that maximizes $f_C(v)$.
	
	\paragraph*{Diameter.} Fix a vertex $v\in V$. As $\cC$ is a sparse cover, there is some $\cC$-cluster $C_v\in\cC$ such that $B_G(v,\frac{\Delta}{\beta})\subseteq C_v$, implying $f_{C_v}(v)> \frac{\Delta}{\beta}$.
    From the other hand, for every $\cC$-cluster $C$ which does not contain $v$, it holds that $f_C(v)\le\delta_C\cdot\frac{\Delta}{\beta}+0\le\frac{\Delta}{\beta}$. It follows that $f_C(v)<f_{C_v}(v)$. Hence $v$ can only join a $\cP$-cluster in $P_C$ that is associated with a $\cC$-cluster $C$ that contains it. In particular, for every $C\in\cC$, $P_C\subseteq C$. As every $\cC$-cluster has diameter at most $\Delta$, we conclude that $\cP$ has diameter at most $\Delta$. Note that the diameter guarantee in $\cP$ is weak, regardless of the diameter guarantee in $\cC$.

    \paragraph*{Padding probability.} Next we analyze the padding probability. We begin with a claim that provides a sufficient condition for a ball to be contained in some $\cP$-cluster.

\begin{claim}\label{claim:PadProperty}
 	Let $v$ be a vertex. If $C$ is a $\cC$-cluster such that $f_C(v) > \max_{C' \neq C}f_{C'}(v) + 2r$ for some $r > 0$, then $B_G(v,r) \subseteq P_C$.
  \end{claim}
\begin{proof}
  For every $u\in B_G(v,r)$, and every center $C'\ne C$ it holds that
\begin{align*}
    f_{C}(u) & =\delta_{C}\cdot\frac{\Delta}{\beta}+\partial_{C}(u)~\ge~\delta_{C}\cdot\frac{\Delta}{\beta}+\partial_{C}(v)-d_{G}(u,v)\\
     & =f_{C}(v)-d_{G}(u,v)~\boldsymbol{>}~f_{C'}(v)+2r-d_{G}(u,v)\\
     & =\delta_{C'}\cdot\frac{\Delta}{\beta}+\partial_{C'}(v)+2r-d_{G}(u,v)\\
     & \ge\delta_{C'}\cdot\frac{\Delta}{\beta}+\partial_{C'}(u)+2r-2\cdot d_{G}(u,v)~\ge~f_{C'}(u)~.
\end{align*}
It follows that $f_{C}(u)>f_{C'}(u)$ for every $C'\ne C$, and hence $u\in P_C$. In particular, $B_G(v,r)\subseteq P_{C}$ as required. 
\end{proof}

Consider some vertex $v\in V$, and parameter $\gamma\le\frac{1}{4}$. We will first argue that the ball $B \coloneqq B_G(v,\gamma\cdot \frac{\Delta}{\beta})$ is fully contained in a single $\cP$-cluster with probability at least $e^{-4\gamma\cdot\lambda}$. The theorem will then follow due to scaling.
Let $\cC=\{C_1,C_2,\dots\}$ be an arbitrarily ordering of the $\cC$-clusters. 
Denote by \EMPH{$\mathcal{F}_i$} the event that $v$ joins the cluster associated with $C_i$, i.e. $v\in P_{C_i}$. 
Denote by \EMPH{$\cQ_i$} the event that $v\in P_{C_i}$, but not all of the vertices in $B$ joined $P_{C_i}$, that is $v\in P_{C_i}$ and $P_{C_i} \cap B\neq B$. 
Let $\EMPH{$\cC_v$}\subseteq\cC$ be the clusters containing $v$ (in the cover). Note that for every $C_i\notin \cC_v$, $\Pr[\cF_i]=\Pr[\cQ_i]=0$.
To prove our assertion, it is enough to show that $\Pr\left[\cup_{i}\cQ_{i}\right]\le 1-e^{-4\gamma\cdot\lambda}$.
Set $\EMPH{$\alpha$}\coloneqq e^{-2\gamma\cdot\lambda}$.
\begin{claim}
	For every $i$, 
	$\Pr\left[\cQ_{i}\right]\le\left(1-\alpha\right)\left(\Pr\left[\mathcal{F}_{i}\right]+\frac{1}{e^{\lambda}-1}\right)$.
\end{claim}
\begin{proof}
	As the order in $\cC$ is arbitrary, assume w.l.o.g. that $i=|\cC|$ and denote $C\coloneqq C_{i}$, $\cQ\coloneqq \cQ_i$, $\mathcal{F}\coloneqq \mathcal{F}_i$, and $\delta\coloneqq \delta_{C_i}$.
    We begin by proving the claim conditioned on the samples of all the clusters other than $C$.
    Specifically, let $X\in [0,1]^{|\cC|-1}$ be the vector where the $j$'th coordinate equals $\delta_{x_j}$. Set 
	\[\EMPH{$\rho_X$}\coloneqq \max\left\{0,\frac{\beta}{\Delta}\cdot
    \left(\max_{C'\ne C}\left\{f_{C'}(v)\right\} -\partial_{C}(v)\right)\right\}.\]
    Note that $\rho_X$ is the minimal value of $\delta$ such that if $\delta>\rho_X$, then $C$ has the maximal value $f_C(v)$, and therefore $v$ will join $P_C$. 
	Note that it is possible that $\rho_X>1$.
	Conditioning on all the other samples having values $X$, and assuming first that $\rho_{X}\le 1$, it holds that
	\[
	\Pr\left[\mathcal{F}\mid X\right]=\Pr\left[\delta>\rho_{X}\right]=\int_{\rho_{X}}^{1}\frac{\lambda\cdot e^{-\lambda y}}{1-e^{-\lambda}}dy=\frac{e^{-\rho_{X}\cdot\lambda}-e^{-\lambda}}{1-e^{-\lambda}}~.
	\]    
	If $\delta>\rho_{X}+2\gamma$, then
\[
    f_{C}(v)~=~\delta_{C}\cdot\frac{\Delta}{\beta}+\partial_{C}(v)~\boldsymbol{>}~(\rho_{X}+2\gamma)\cdot\frac{\Delta}{\beta}+\partial_{C}(v)~\ge~2\gamma\cdot\frac{\Delta}{\beta}+\max_{C'\ne C}f_{C'}(v)~.
\]
	In particular, by \Cref{claim:PadProperty} the ball $B$ will be contained in $C$.
	We conclude
	\begin{align*}
		\Pr\left[\cQ\mid X\right] & \le\Pr\left[\rho_{X}\le\delta\le\rho_{X}+2\gamma\right]\\
		& =\int_{\rho_{X}}^{\max\left\{ 1,\rho_{X}+2\gamma\right\} }\frac{\lambda\cdot e^{-\lambda y}}{1-e^{-\lambda}}dy\\
		& \le\frac{e^{-\rho_{X}\cdot\lambda}-e^{-\left(\rho_{X}+2\gamma\right)\cdot\lambda}}{1-e^{-\lambda}}\\
		& =\left(1-e^{-2\gamma\cdot\lambda}\right)\cdot\frac{e^{-\rho_{X}\cdot\lambda}}{1-e^{-\lambda}}\\
		& =\left(1-\alpha\right)\cdot\left(\Pr\left[\mathcal{F}\mid X\right]+\frac{1}{e^{\lambda}-1}\right)~.
	\end{align*}
	Note that if $\rho_{X}> 1$ then $\Pr\left[\cQ\mid X\right]=0\le \left(1-\alpha\right)\cdot\left(\Pr\left[\mathcal{F}\mid X\right]+\frac{1}{e^{\lambda}-1}\right)$ as well.
	Denote by $h$ the density function of the distribution over all possible values of $X$. Using the law of total probability, we can bound the probability of $\cQ$:
	\begin{align*}
		\Pr\left[\cQ\right] & =\int_{X}\Pr\left[\cQ\mid X\right]\cdot h(X)~dX\\
		& \le\left(1-\alpha\right)\cdot\int_{X}\left(\Pr\left[\mathcal{F}\mid X\right]+\frac{1}{e^{\lambda}-1}\right)\cdot h(X)~dX\\
		& =\left(1-\alpha\right)\cdot\left(\Pr\left[\mathcal{F}\right]+\frac{1}{e^{\lambda}-1}\right)~.
	\end{align*}
\end{proof}

Next, we bound the probability that the ball $B$ is cut. 
\begin{align}	\Pr\left[\cup_{i}\cQ_{i}\right]=\sum_{C_{i}\in\cC_{v}}\Pr\left[\cQ_{i}\right] & \le\left(1-\alpha\right)\cdot\sum_{C_{i}\in\cC_{v}}\left(\Pr\left[\mathcal{F}_{i}\right]+\frac{1}{e^{\lambda}-1}\right)\nonumber\\
	& \le\left(1-e^{-2\gamma\cdot\lambda}\right)\cdot\left(1+\frac{s}{e^{\lambda}-1}\right)\nonumber\\
	& \le\left(1-e^{-2\gamma\cdot\lambda}\right)\cdot\left(1+e^{-2\gamma\cdot\lambda}\right)=1-e^{-4\gamma\cdot\lambda}~,\label{eq:cutProb}
\end{align}
where the last inequality follows as
$e^{-2\gamma\lambda}=\frac{e^{-2\gamma\lambda}\left(e^{\lambda}-1\right)}{e^{\lambda}-1}\ge\frac{e^{-2\gamma\lambda}\cdot e^{\lambda-1}}{e^{\lambda}-1}\ge\frac{e^{\frac{\lambda}{2}-1}}{e^{\lambda}-1}=\frac{s}{e^{\lambda}-1}$, where we used that $\gamma\in(0,\frac14]$, and $\lambda=2+2\ln s$.

We conclude that our distribution indeed produces a weak $(4\beta\cdot\lambda,\frac{1}{4\beta},\Delta)$-padded decomposition. Indeed, we already established that the diameter is at most $\Delta$. Next, fix $\gamma\le\frac{1}{4\beta}$, and $v\in V$. Denote by $P(v)$ the $\cP$-cluster containing $v$. Then
\[
\Pr_{\cP}\left[B_{G}(v,\gamma\cdot\Delta)\subseteq P(v)\right]=\Pr_{\cP}\left[B_{G}(v,\beta\cdot\gamma\cdot\frac{\Delta}{\beta})\subseteq P(v)\right]\ge e^{-4\cdot\beta\cdot\gamma\cdot\lambda}~,
\]
where we used inequality (\ref{eq:cutProb}) w.r.t. $\beta\cdot\gamma$, and the fact that $\beta\cdot\gamma\le\frac14$. The theorem now follows.

\end{proof}

\section{Applications}\label{sec:apps}
Our \Cref{thm:padded} and \Cref{thm:cover} have numerous algorithmic applications. We highlight some of them:%in the following list:
\begin{enumerate}
\item \textbf{Multi-Commodity Max-Flow/Min-Cut Gap}: Here we are given undirected graph $G=(V,E)$ with capacity function $c:E\rightarrow\R_{\ge0}$ over the edges, and $k$ demand pairs $(s_i,t_i)$. There are two different versions of the problem.
    \begin{itemize}
        \item \textbf{Maximum throughput/Minimum Multicut}:
        In the maximum throughput version the goal is to send the maximum total amount of flow (denoted $F_{\rm tp}$) between the demand pairs (while respecting the capacity constrains). 
        In the minimum multicut problem the goal is to delete edges  of minimum total capacity (denoted $C_{\rm multi}$) so that to separate every demand pair. 
        The maximum throughput is bounded by the capacity of the minimum multicut, that is $F_{\rm tp}\le C_{\rm multi}$.
        The ratio $\nicefrac{C_{\rm multi}}{F_{\rm tp}}$ is called the multiflow-multicut gap.
        Using \cite{TV93} and \Cref{thm:padded}, it follows that in $K_r$-minor free graph the multiflow-multicut gap is at most $O(\log r)$, which is also tight \cite{GVY96} (as there is an $\Omega(\log n)$ lower bound for general graphs).
        The previously known upper bound was $O(r)$ \cite{TV93,AGGNT19}.

        \item \textbf{Maximum  concurrent flow/Cut Ratio}: Here in addition each demand pair $(s_i,t_i)$ has demand $D_i$. 
        The goal is to route the maximum fractional flow $f\cdot D_i$ between all the demand pairs simultaneously.
        The minimum cut ratio is $R=\min_{U\subset V}\frac{C(U,\bar{U})}{D(U,\bar{U})}=\min_{U\subset V}\frac{\sum_{e\in E\cap(U\times\bar{U})}c(e)}{\sum_{(s_{i},t_{i})\in(U\times\bar{U})\cup(\bar{U}\times U)}D_{i}}$.
        The maximum  concurrent flow is bounded by the minimum cut ratio, that is $f\le R$.
        For the case of $K_r$-minor free graph with uniform demands (where there is a unit demand between every vertex pair), \cite{KPR93} and \Cref{thm:padded} provide an $O(\log r)$ bound on the minimum cut ratio.
        An $O(\log r)$ approximation to the sparsest cut with uniform demands follows \cite{KPR93, Rab03}.
        Both improving a previous $O(r)$ upper bounds \cite{KPR93,Rab03,AGGNT19}.
    \end{itemize}

\item \textbf{Flow sparsifier}: Given an edge-capacitated graph $G=(V,E,c)$ and a subset of terminals $K\subseteq V$, a flow sparsifier with quality $q\ge 1$ is a graph $H$ over the terminal set $K$, such that (a) any feasible $K$-flow in $G$ can be feasibly routed in $H$, and (b) any feasible $K$-flow in $H$ can be routed in $G$ with congestion $\rho$.
Based on \cite{EGKRTT14} and \Cref{thm:padded}, every $K_r$ minor free graph admits (efficiently commutable) flow sparsifier of quality $O(\log r)$, improving a previous $O(r)$ quality flow sparsifier \cite{EGKRTT14,AGGNT19}. 

\item \textbf{Sparse partition}: A graph admits $(\rho,\tau)$-sparse partition scheme if for every $\Delta>0$ there is a partition into clusters of diameter $\Delta$ such that every ball of radius $\frac\Delta\rho$ intersects at most $\tau$ clusters. Based on \cite{JLNRS05,Fil24scattering} and \Cref{thm:cover} every $K_r$-minor free graph admits an 
$(O(1),O(r^4))$-sparse partition scheme.
This exponentially improves the sparsity parameter compared to the $(4+\eps,O(\frac1\eps)^4)$-sparse partitions from \cite{Fil24sparse}.

\item \textbf{Steiner point removal}: Here we are given a weighted graph $G=(V,E,w)$ and a subset of terminals $K\subseteq V$. The goal is to construct a weighted minor $H$ of $G$ with $K$ as its vertex set while preserving all terminal pairwise distances up to a small multiplicative distortion. 
Using \cite{EGKRTT14} and \Cref{thm:padded}, for every $K_r$-minor free graph, there is a distribution over dominating minors with expected distortion $\tilde{O}(\log r)$, improving a previous $\tilde{O}(r)$ bound \cite{EGKRTT14,AGGNT19}.

\item \textbf{$\boldsymbol{0}$-Extension and MultiWay Cut}: 
In the $0$-Extension problem, the input is a set $X$, a terminal set $K\subseteq X$, a metric $d_K$ over the terminals, and an arbitrary cost function $c: {X\choose 2}\rightarrow \mathbb{R}_+$.
The goal is to find a \emph{retraction} $f:X\rightarrow K$ 
that minimizes $\sum_{\{x,y\}\in {X\choose 2}} c(x,y)\cdot d_K(f(x),f(y))$.
A retraction is a surjective function $f:X\rightarrow K$ that satisfies $f(x)=x$ for all $x\in K$. 
An important special case is the Multiway Cut problem, where the goal is to cut a minimum number of edges to separate $k$ terminals into $k$ disjoint sets.
For the case where the metric $(K,d_K)$ over the terminals is induced by the shortest path metric of a $K_r$-minor free graph,
\cite{LN05} and \Cref{thm:padded} provide an $O(\log r)$ approximation algorithm for the $0$-extension problem (see also \cite{AFHKTT04,CKR04}).
This improved a previous $O(r)$ approximation \cite{LN05,AGGNT19}.

\item \textbf{Lipschitz Extension}: 
For a function $f:X\rightarrow Y$ from a metric spaces $(X,d_X)$ into a Banach space $Y$, set $\|f\|_{{\rm Lip}}=\sup_{x,y\in X}\frac{d_Y(f(x),f(y))}{d_X(x,y)}$ to be the Lipschitz parameter of the function. In the Lipschitz extension problem, we are given a map $f:Z\rightarrow Y$ from a subset $Z$ of $X$. The goal is to extend $f$ to a function $\tilde{f}$ over the entire space $X$, while minimizing $\|\tilde{f}\|_{{\rm Lip}}$ as a function of $\|f\|_{{\rm Lip}}$.
Suppose that $(X,d_X)$ is the shortest path metric of a $K_r$ minor free graph. Then from \cite{LN05} and \Cref{thm:padded} it follows that there is always an extension with  Lipschitz parameter $\|\tilde{f}\|_{{\rm Lip}}\le O(\log r)\cdot \|f\|_{{\rm Lip}}$ (previously it was known that extension with parameter $O(r)\cdot \|f\|_{{\rm Lip}}$ is possible \cite{LN05,AGGNT19}).

\item \label{app:EmbeddingLp}\textbf{Metric embedding into $\boldsymbol{\ell_p}$ spaces}: Given an $n$-vertex graph $G$, here the goal is to embed the vertices of $G$ into $\ell_p$ while preserving pairwise distances up to a small multiplicative distortion. 
From \cite{KLMN04} and \Cref{thm:padded}, it follows that every $K_r$-minor free graph can be embedded into $\ell_p$ with distortion $O((\log r)^{1-\frac1p}\cdot(\log n)^{\frac1p})$ (improving over  $O(r^{1-\frac1p}\cdot(\log n)^{\frac1p})$\cite{KLMN04,AGGNT19}).

\item \textbf{Metric embedding into $\boldsymbol{\ell_\infty}$:} Every metric embeds into $\ell_\infty$  isometrically, thus here we study the trade-off between distortion and dimension.
From \cite{Fil24sparse} and \Cref{thm:cover}, it follows that every $K_r$-minor free graph can be embedded into $\ell_\infty$ with distortion $O(1)$ and $O(r^4\cdot\log n)$ dimensions (improving over  $O(1)$ distortion and $O(1)^r\cdot\log n$ dimensions \cite{Fil24sparse}).

\item \textbf{Average distortion, vertex cuts, and treewidth approximation:} Here we sketch three interconnected applications.
    \begin{itemize}
        \item \textbf{Average distortion embedding into the line:} Given an undirected weighted graph $G = (V,E,w)$, our goal is to find a non-expanding embedding $f:V \to \R$ (that is, $|f(x) - f(y)| \le \dist_G(x,y)$ for all $x,y \in V$) that minimizes the average distortion $\frac{\sum_{x, y\in V}\dist_G(x,y)}{\sum_{x, y\in V} |f(x) - f(y)|}$. Following Rabinovich \cite{Rab03}, our \Cref{thm:padded} implies that every $K_r$-minor-free graph can be embedded into the line with $O(\log r)$ average distortion (improving over $O(r)$ \cite{Rab03, AGGNT19}).

        \item \textbf{Min-ratio vertex separator:} A vertex separator is a partition of $V$ into three sets $(A, S, B)$ such that deleting $S$ from $G$ separates the graph into disconnected pieces $G[A]$ and $G[B]$. Let $\pi:V \to \R_{> 0}$ be a function defining weights on the vertices of $G$; for any subset $S \subseteq V$ we write $\pi(S)$ to mean $\sum_{v \in S}\pi(v)$. The sparsity of a vertex separator $(A, S, B)$ is $\frac{\pi(S)}{\pi(A \cup S) \cdot \pi(B \cup S)}$. Feige, Hajiaghayi, and Lee \cite{FHL08} showed that average-distortion embedding into the line can be used to find an approximate minimum-sparsity vertex separator. Our \Cref{thm:padded} (together with the arguments of \cite{FHL08} Theorem 4.2) yields an $O(\log r)$ approximation. Previously, an $O(r)$ approximation and an $O(\sqrt{\mathrm{opt}})$ approximation were known \cite{AGGNT19, FHL08}.

        \item \textbf{Balanced separators and Treewidth:} Let $W \subseteq V$ be an arbitrary subset of vertices. For any $\delta \in (0,1)$, a $\delta$-balanced vertex separator with respect to $W$ is a vertex separator $(A,S,B)$ such that $|A \cap W|$ and $|B \cap W|$ are at most $\delta |W|$. Given a subset $W \subseteq V$, we would like to find a $\delta$-balanced vertex separator $(A,S,B)$ that minimizes $|S|$. Following \cite{LR99, FHL08}, our $O(\log r)$-approximation to min-sparsity vertex separator can be used to find a $3/4$-balanced separator whose size is within $O(\log r)$ of the optimal $2/3$-balanced separator. An $O(\log r)$-approximation to treewidth follows \cite{BGHK95, FHL08}. Previously, an $O(r)$ approximation and an $O(\sqrt{\mathrm{opt}})$ approximation were known \cite{AGGNT19, FHL08}.
        
\end{itemize}

\end{enumerate}

\section{Concluding Remarks}
\paragraph*{General Graphs.} The main result of the paper is the construction of optimal padded decompositions for $K_r$ minor free graphs with padding parameter $O(\log r)$.
As the padding parameter of general $n$ vertex graphs is $O(\log n)$, it will be beneficial to use our decomposition even if a given graph only excludes a rather large minor (say $K_{2^{O(\sqrt{\log n})}}$).
However, given a graph $G$, unless $r$ is a constant \cite{KPS24}, the best known approximation factor for the minimum $r$ such that $G$ is $K_r$ minor free is $O(\sqrt{n})$ \cite{ALW07}.
Nevertheless, given a $K_r$ minor free graph $G$ we can sample a decomposition with parameter $O(\log r)$, without any knowledge of $r$! 

The details are as follows: the first step in our algorithm is to construct a buffered cop decomposition (\Cref{def:buffered-cop}, \cite{CCLMST24}).
In Chang \etal's \cite{CCLMST24} construction, one provides a buffer parameter $\gamma$, and if the graph is $K_r$ minor free, then each supernode will have radius $r\cdot \gamma$. A prior knowledge of $r$ is not required. 
(This is implicit in \cite{CCLMST24}. They phrase their algorithm as taking input $\Delta$ and $r$, but they use these parameters only to define the buffer parameter $\gamma \gets \Delta/r$.)
One can try different choices of $\gamma$ until all supernodes will have radius $\Delta$ (and thus the buffer will be at least $\frac{\Delta}{r}$). 
Then, given the buffered cop decomposition we construct sparse cover (\Cref{thm:BufferedToSparseCover}), and finally we sample padded decomposition from the sparse cover (\Cref{thm:coversToPaddedDecomp}), obtaining the desired padded decomposition.

\paragraph*{Running Time.} 
While the running time of all the algorithms used and presented in the paper is polynomial, it is not explicitly stated. In particular, the running time of the buffered cop decomposition \cite{CCLMST24} is not explicit.
Nevertheless, after a polynomial time preprocessing, one can sample a padded decomposition using our \Cref{thm:padded} in $O(n\cdot r^4)$ time.
Indeed, for the preprocessing, construct an $(O(1),O(r^4),\Delta)$-sparse cover $\cC$ (\Cref{thm:cover}). In addition, for every cluster $C\in\cC$ and $v\in C$, compute the distance to the boundary $\partial_{C}(v)$.
Then, to sample a padded decomposition, we sample the random shifts $\{\delta_C\}_{C\in\cC}$, and assign each vertex $v$ to the cluster $C$ (containing $v$) that minimizes the function $f_C(v)$. As the sparsity of the cover is $O(r^4)$, the running time is $O(n\cdot r^4)$.

\paragraph*{Open Questions.} We list several open problems following our paper:
\begin{enumerate}
    \item Strong diameter padded decomposition:\label{question:StrongPadded} Our padded decompositions (\Cref{thm:padded}) have only weak diameter guarantee. For strong diameter, the state of the art remains an $O(r)$ padding parameter \cite{Fil19padded}. 
    In fact, even for treewidth $\tw$ graphs, nothing better is known.
    Closing this exponential gap is the main open question left by this paper.
    
    \item Strong diameter sparse covers:\label{question:StrongSparseCover} In \Cref{thm:cover} we constructed sparse covers with weak diameter, constant padding and $\poly(r)$ sparsity. From the other hand, for strong diameter, the best known sparse cover with constant padding has $\exp(r)$ sparsity \cite{Fil24sparse}. Closing this gap is another interesting question. 
    
    \item Reduction from strong sparse covers to strong padded decomposition: In \Cref{thm:coversToPaddedDecomp} even if one plugs in sparse cover with strong diameter, the resulting padded decomposition will only have weak diameter (in particular, a positive answer to \Questionref{question:StrongSparseCover} will not imply a similar improvment for \Questionref{question:StrongPadded}).
    We ask whether there is a similar reduction that given a sparse cover with strong diameter will produce strong diameter padded decomposition, or is there a separation between the two?
    
    \item Buffered cop decomposition with constant buffer: Given a $K_r$ minor free graph $G$, Chang \emph{et~al.} \cite{CCLMST24} constructed a $(\Delta, \Delta/r, r-1)$-buffered cop decomposition. 
    Is it possible to construct a $(\Delta, \Delta/O(1), r-1)$-buffered cop decomposition?
    Note that based on the arguments in \cite{Fil24sparse}, a positive resolution to this question will also provide a positive answer to \Questionref{question:StrongSparseCover}.

    \item Optimal padding for the sparse cover: In our sparse cover (\Cref{thm:cover}) we get padding $8+\eps$. Note that the best padding one can hope for is $4$. Indeed, consider an unweighted star where we subdivided each edge. Any cover with padding strictly better than $4$ will have sparsity $\Omega(n)$. 
    Filtser \cite{Fil24sparse} constructed a sparse cover with padding $4+\eps$ but with $\exp(r)$ sparsity. Is it possible to get padding $4+\eps$ with polynomial sparsity?
\end{enumerate}

\section*{Acknowledgments}
The first author would like to thank Hsien-Chih Chang for many useful discussions on cop decomposition. The second author would like to thank Ofer Neiman, D Ellis Hershkowitz, Jason Li, Nikhil Kumar, and Davis Issac for helpful discussions.

\small
\bibliographystyle{alphaurl}
\bibliography{main,LSObib}

\end{document}